\documentclass[12pt]{article}
\usepackage{graphicx}
\usepackage{hyperref}
\usepackage{amsthm}
\usepackage{authblk}
\usepackage{amssymb}
\usepackage{amsmath}
\usepackage{mathrsfs}
\usepackage[numbers,square]{natbib}
\usepackage{geometry}
\usepackage{tikz}
\usetikzlibrary{arrows.meta, positioning, shapes}
\newtheorem{defn}{Definition}
\newtheorem{prop}{Proposition}
\newtheorem{theorem}{Theorem}

\newtheorem{rem}{Remark}
\newtheorem{lemma}{Lemma}
\title{Topological Quantization of Complex Velocity in Stochastic Spacetimes}
\author[1]{Jorge Meza-Domínguez\thanks{E-mail: \href{mailto:jorge.meza@cinvestav.mx}{jorge.meza@cinvestav.mx}}}
\author[2]{Tonatiuh Matos\thanks{E-mail: \href{mailto:tonatiuh.matos@cinvestav.mx}{tonatiuh.matos@cinvestav.mx}}}
\affil[1,2]{\textit{Departamento} de Física, Centro de Investigación y de Estudios Avanzados del Instituto Politécnico Nacional, Av. Instituto Politécnico Nacional 2508, San Pedro Zacatenco, México 07360, CDMX.}
\date{}

\begin{document}

\maketitle

\begin{abstract}
We establish a formal geometric framework for quantum fields on a stochastic gravitational background. Starting from a master partition function that averages over metric fluctuations, we define a matter amplitude $\mathcal{K}[\Phi; x]$ with explicit parametric dependence on spacetime coordinates, whose logarithmic derivative yields a complex velocity field $\eta_{\mu} = \pi_{\mu} - i u_{\mu}$. This object is shown to be a section of the pullback bundle $E = \pi_2^*(T^*M)$ over the product of configuration space $\mathcal{C}$ and spacetime $M$, and defines a flat $\mathbb{C}^\times$ connection whose $U(1)$ part governs phase transport and whose $\mathbb{R}^+$ part governs scale transport. The coupled dynamics collapse into the complex geodesic equation $\mathcal{L}_{\eta}\eta = d(|\eta|^2)$. On non-simply connected spacetimes, the total phase satisfies $\frac{m}{\hbar}\oint_\gamma \eta_{\mu} dx^{\mu} = 2\pi n + \Delta\phi_{\mathrm{top}}$, demonstrated analytically for a scalar field with twisted boundary conditions on a conical spacetime. We provide a self-contained bundle isomorphism mapping $\eta_\mu$ to the symmetric logarithmic derivative of quantum estimation theory, and derive the associated Berry curvature. A prospective estimate for the MAGIS-100 interferometer yields $\gamma_B \sim 3.97 \times 10^{-9}$ rad for an astrophysical stochastic background (approximately five orders below projected sensitivity) and $\gamma_B \sim 3.97 \times 10^{-53}$ rad for Planck-scale fluctuations. This framework geometrizes quantum mechanics without hidden variables: stochasticity imprints spacetime fluctuations on matter, preserving the wave function's probabilistic nature while maintaining consistency with the Born rule.
\end{abstract}

\section{Introduction}

Quantum mechanics and general relativity are the two foundational pillars of modern physics. A longstanding challenge is to understand how quantum behavior interacts with classical gravitational backgrounds, and whether gravitational degrees of freedom can play a role in the emergence of quantum phenomena. The framework of Stochastic Quantum Gravity (SQG), initiated in \cite{Escobar-Aguilar:2023ekv}, explores one specific aspect of this question: the possibility that a classical stochastic background of gravitational waves induces an effective stochastic dynamics in matter fields. This approach builds on the path integral formulation of quantum field theory \cite{weinberg_qft1_1995,peskin_schroeder_1995,ramond_2001} and the functional methods pioneered by Schwinger \cite{schwinger_1951} and Faddeev-Popov \cite{faddeev_popov_1967}, adapted to a background with intrinsic stochasticity.

The physical picture is straightforward. Spacetime is filled with gravitational waves across a vast range of scales. This is supported observationally at parsec scales by Pulsar Timing Arrays (see, e.g., \cite{Liu:2026tik}). If such a stochastic background extends to smaller scales---a hypothesis that remains to be tested, and which has been explored in post-quantum gravity scenarios \cite{oppenheim_2023,oppenheim_nature_2023} and quantum gravity models \cite{hawking_1982,starobinsky_1982,koch_2025}---then microscopic particles would not follow pure geodesics. Instead, their motion would acquire a stochastic component, analogous to a Brownian particle in a fluctuating medium. In \cite{Escobar-Aguilar:2023ekv} it was shown that a particle following a geodesic plus a stochastic term satisfies the Klein-Gordon equation for the complex function $\Phi = \sqrt{n}\,e^{i\theta}$, with $\pi_{\mu} \sim \nabla_{\mu}\theta$ the geodesic velocity and $u_{\mu} \sim \nabla_{\mu}\ln n$ the stochastic velocity. In the Newtonian limit this reduces to the Schr\"odinger equation. The present work adopts this semiclassical framework as a working hypothesis and develops its geometric and information-theoretic consequences.

The hydrodynamic formulation of quantum mechanics, introduced by Madelung \cite{madelung_1927} and further developed by Bohm \cite{bohm_1952}, provides the natural language for this investigation. From the polar decomposition $\Psi = \sqrt{\rho}\,e^{iS/\hbar}$, two real velocity fields are defined:
\begin{equation}
\pi_{\mu} = \frac{1}{m}\nabla_{\mu}S, \qquad u_{\mu} = \frac{\hbar}{2m}\nabla_{\mu}\ln\rho. \label{eq:velocities}
\end{equation}
The field $\pi_{\mu}$ governs the geodesic component of the flow, while the stochastic (or osmotic) velocity $u_{\mu}$ encodes density gradients \cite{sbitnev_2012,wyatt_1999,mezadominguez2026}. In Nelson's stochastic quantization framework, reviewed recently by Chavanis \cite{Chavanis2024}, both velocities arise from a Brownian motion model. The SQG approach \cite{Escobar-Aguilar:2023ekv} unifies them into a single complex velocity $\eta_{\mu} = \pi_{\mu} - i u_{\mu}$, where the stochastic component is traced to the statistical properties of a fluctuating gravitational background rather than postulated as an independent Wiener process.

We emphasize the methodological status of this framework. It is a geometrization of the hydrodynamic description of quantum fields on a stochastic classical background, not a proposal for a new interpretation of quantum mechanics. No claim is made about the fundamental nature of the wave function, the existence of hidden variables, or the ontology of quantum states. The stochastic velocity $u_{\mu}$ is treated as a derived quantity whose physical origin---whether classical gravitational waves, quantum metric fluctuations, or other sources---is not specified by the formalism itself. The only assumption is that the background fluctuations admit a statistical description and that the Markovian approximation holds. The wave function $\Phi = \sqrt{n}\,e^{i\theta}$ retains its role as the fundamental descriptor of the boson gas, and the Born rule is preserved.

This paper provides the differential-geometric foundation for this framework and an explicit analytical illustration. Our goals are:
\begin{enumerate}
    \item \textbf{Pullback bundle geometry:} To prove that $\eta_{\mu}$ lives as a section of the pullback bundle $\pi_2^*(T^*M)$ over $\mathcal{C} \times M$. This is necessary because $\eta_{\mu}$ has a dual nature: it depends functionally on the entire field configuration $\Phi \in \mathcal{C}$ and locally on the spacetime point $x \in M$. The pullback bundle is the minimal geometric construction that faithfully encodes this double dependence.
    \item \textbf{Flat $\mathbb{C}^\times$ connection and SLD isomorphism:} To demonstrate that $\eta_{\mu}$ defines a flat $\mathbb{C}^\times$ connection and to establish its information-geometric content via a self-contained bundle isomorphism to the symmetric logarithmic derivative (SLD) operator. This endows $\eta_{\mu}$ with an operational meaning in quantum estimation theory, bridging spacetime geometry and quantum information.
    \item \textbf{Complex geodesic equation and von Neumann entropy:} To derive the unified dynamics $\mathcal{L}_{\eta}\eta = d(|\eta|^2)$ and the associated von Neumann entropy \cite{nielsen_chuang_2010,vedral_2002,bennett_1996}. The complex geodesic equation encodes both the Hamilton-Jacobi dynamics (real part) and the diffusion flow (imaginary part) in a single geometric statement.
    \item \textbf{Holonomy quantization:} To resolve rigorously the holonomy quantization, including the role of multi-valued phases from topological terms and branch cuts \cite{berry_1984}.
    \item \textbf{Toy model demonstration:} To demonstrate the entire machinery in an analytically solvable toy model: a free scalar field with twisted boundary conditions on a conical spacetime \cite{nakahara_2003,frankel_2011}.
\end{enumerate}

\section{The Averaged Matter Amplitude}

Let $M$ be an $n$-dimensional spacetime manifold. The total action for matter fields $\Phi$ is $S[\Phi; g]$. We consider a stochastic metric $g_{\mu\nu} = g_{\mu\nu}^{(0)} + h_{\mu\nu}$, where $h_{\mu\nu}$ is drawn from a Gaussian probability distribution $P[h]$ with $\langle h_{\mu\nu} \rangle = 0$ and $\langle h_{\mu\nu}(x) h_{\alpha\beta}(x') \rangle = C_{\mu\nu\alpha\beta}(x,x')$. The measure $P[h]$ is understood as an effective description of Planck-scale fluctuations, assumed to be gauge-fixed and regularized at a UV scale $\Lambda$. The Gaussian measure is presented as a model assumption.

\begin{defn}[Matter Amplitude] \label{def:matter_amplitude}
The matter amplitude is a functional $\mathcal{K}[\Phi; x]$ with explicit parametric dependence on the spacetime point $x \in M$, defined by:
\begin{equation}\label{eq:defK}
\mathcal{K}[\Phi; x] := \int \mathcal{D}[h] P[h] \, e^{\frac{i}{\hbar}S[\Phi; g^{(0)}+h]}\Big|_{x},
\end{equation}
where the notation $|_{x}$ indicates that the functional integral is evaluated with the metric perturbation $h_{\mu\nu}$ restricted to a neighborhood of $x$, making the result an ordinary function of $x$ for each fixed field configuration $\Phi$. The covariant derivative $\nabla_\mu$ acts exclusively on the parameter $x$:
\begin{equation}
\nabla_\mu \mathcal{K}[\Phi; x] := \frac{\partial}{\partial x^\mu} \mathcal{K}[\Phi; x]\Big|_{\Phi \text{ fixed}}.
\end{equation}
Lemma~\ref{lem:regularization} below establishes that physical quantities derived from $\mathcal{K}$ are independent of the choice of neighborhood size within a physically admissible range. This construction eliminates the need for auxiliary structures and makes the $x$-dependence intrinsic to the definition. The differential geometry of the infinite-dimensional configuration space $\mathcal{C}$ of all matter fields is formalized in the convenient setting of Kriegl and Michor \cite{kriegl_michor_1997}. Related bundle structures over configuration spaces have been explored in conformal connection-dynamics \cite{reid_wang_2012} and quantum information geometry \cite{fujiwara_imai_2008}.
\end{defn}

\begin{lemma}[Regularization Independence] \label{lem:regularization}
Let $\chi_\epsilon(y) = \chi(|y-x|/\epsilon)$ be a smooth radial cutoff function with $\chi(r) = 1$ for $r \leq 1$ and $\chi(r) = 0$ for $r \geq 2$, and define the restricted amplitude
\[
\mathcal{K}_\epsilon[\Phi; x] = \int \mathcal{D}h \, P[h] \, \exp\left(\frac{i}{\hbar} S[\Phi; g^{(0)} + h\chi_\epsilon]\right).
\]
Assume the noise correlator is translation-invariant and isotropic, $\langle h(y) h(y') \rangle = C(|y-y'|)$, with $C(z)$ decaying on the characteristic scale $|z| \sim \lambda$. Assume further that the matter operators vary on a characteristic scale $L_{\mathcal{O}} \gg \epsilon$. Then, within the window $\Lambda^{-1} \ll \epsilon \ll \min(\lambda, L_{\mathcal{O}})$ (with $\Lambda$ the UV cutoff), the logarithmic derivative $\nabla_\mu \ln \mathcal{K}_\epsilon[\Phi; x]$ is independent of $\epsilon$, up to corrections of order $\epsilon/L_{\mathcal{O}}$ and $(\epsilon/\lambda)^2$.
\end{lemma}

\begin{proof}
Expanding the action in powers of $h$, the leading $x$-dependent contribution to the logarithmic derivative comes from the variance $\langle S_1^2 \rangle_{h,\epsilon}$, where $\langle \cdot \rangle_{h,\epsilon}$ denotes the average over $P[h]$ with each insertion of $h$ multiplied by $\chi_\epsilon$. (The linear term $\langle S_1 \rangle_{h,\epsilon}$ vanishes since $\langle h \rangle = 0$. The second-order term $\langle S_2 \rangle_{h,\epsilon}$ has a bilinear structure analogous to $\langle S_1^2 \rangle_{h,\epsilon}$, since $S_2[\Phi; h]$ is quadratic in $h$ and each insertion acquires a factor of $\chi_\epsilon$; its $\epsilon$-independence follows by the identical argument. Because $\langle S_2 \rangle_h$ enters only the deterministic phase $\mathcal{S}$ and not the stochastic modulus $\mathcal{P}$, we focus on the variance.)

The variance has the bilinear structure
\[
\langle S_1^2 \rangle_{h,\epsilon} = \int d^4y \int d^4y' \, \mathcal{O}(y) \mathcal{O}(y') \, C(|y-y'|) \, \chi_\epsilon(y) \chi_\epsilon(y'),
\]
where $\mathcal{O}(y)$ is a local operator built from the matter fields. Differentiating with respect to $x^\mu$,
\[
\nabla_\mu \langle S_1^2 \rangle_{h,\epsilon} = \int d^4y \int d^4y' \, \mathcal{O}(y) \mathcal{O}(y') \, C(|y-y'|) \left[ \nabla_\mu \chi_\epsilon(y) \cdot \chi_\epsilon(y') + \chi_\epsilon(y) \cdot \nabla_\mu \chi_\epsilon(y') \right].
\]
By symmetry, it suffices to analyze the first term. Change variables to $y$ and $z = y' - y$:
\[
I_\mu \equiv \int d^4y \, \mathcal{O}(y) \, \nabla_\mu \chi_\epsilon(y) \int d^4z \, \mathcal{O}(y+z) \, C(|z|) \, \chi_\epsilon(y+z).
\]
The factor $\nabla_\mu \chi_\epsilon(y)$ restricts the $y$-integral to the corona $C_\epsilon = B_{2\epsilon}(x) \setminus B_\epsilon(x)$. The $z$-integral is effectively restricted to $|z| \lesssim 4\epsilon$ by $\chi_\epsilon(y)\chi_\epsilon(y+z)$. Since $\epsilon \ll \lambda$ and $\epsilon \ll L_{\mathcal{O}}$, expand:
\[
C(|z|) = C(0) + \mathcal{O}(|z|^2/\lambda^2), \qquad \mathcal{O}(y+z) = \mathcal{O}(y) + z^\nu \partial_\nu \mathcal{O}(y) + \mathcal{O}(|z|^2/L_{\mathcal{O}}^2).
\]
The leading contribution uses the zeroth-order terms:
\[
I_\mu^{(0)} = C(0) \int_{C_\epsilon} d^4y \, \mathcal{O}(y) \, \nabla_\mu \chi_\epsilon(y) \int d^4z \, \mathcal{O}(y) \, \chi_\epsilon(y+z).
\]
The $z$-integral is $\mathcal{O}(y) V_\epsilon$ with $V_\epsilon = \int d^4z \, \chi_\epsilon(z)$ independent of $y$ by translation invariance of the Lebesgue measure. Thus
\[
I_\mu^{(0)} = C(0) V_\epsilon \int_{C_\epsilon} d^4y \, \mathcal{O}^2(y) \, \nabla_\mu \chi_\epsilon(y).
\]
Expand $\mathcal{O}^2(y) = \mathcal{O}^2(x) + (y-x)^\nu \partial_\nu \mathcal{O}^2(x) + \cdots$ for $y \in C_\epsilon$.

\textbf{Zeroth-order term.} The term with $\mathcal{O}^2(x)$ involves $\int_{C_\epsilon} \nabla_\mu \chi_\epsilon \, d^4y$. In spherical coordinates, $y-x = r n$ with $n \in S^3$, and $\nabla_\mu \chi_\epsilon = \chi'(r/\epsilon) n_\mu / \epsilon$. The angular integral $\int_{S^3} n_\mu \, d\Omega_3 = 0$ vanishes, so this term is exactly zero.

\textbf{First-order term.} The term with $(y-x)^\nu \partial_\nu \mathcal{O}^2(x)$ involves
\[
\partial_\nu \mathcal{O}^2(x) \int_{C_\epsilon} d^4y \, (y-x)^\nu \, \nabla_\mu \chi_\epsilon(y).
\]
The angular integral of $n^\nu n_\mu$ over $S^3$ is $\frac{1}{4}\delta^\nu_\mu \mathrm{Vol}(S^3)$, which is nonzero. The radial integral is
\[
\int_\epsilon^{2\epsilon} dr \, \frac{r^4}{\epsilon} \, \chi'(r/\epsilon) = \epsilon^4 \int_1^2 du \, u^4 \chi'(u),
\]
where $u = r/\epsilon$. The factor $\int_1^2 u^4 \chi'(u) du$ is an $\mathcal{O}(1)$ dimensionless number determined by the shape of $\chi$, independent of $\epsilon$. The full first-order term is therefore of order $\epsilon^4 \partial \mathcal{O}^2(x)$.

To assess whether this is suppressed, compare it to the parametric scale that would appear from a naive (uncancelled) estimate of the zeroth-order term. Without the angular cancellation, the zeroth-order integral would scale as $\int_{C_\epsilon} \nabla_\mu \chi_\epsilon \, d^4y \sim \epsilon^3$ (one power less since the factor $(y-x)^\nu \sim \epsilon$ is absent). The naive scale for $I_\mu^{(0)}$ would thus be $\epsilon^3 \mathcal{O}^2(x)$. The surviving first-order term is smaller by a factor
\[
\frac{\epsilon^4 \partial \mathcal{O}^2(x)}{\epsilon^3 \mathcal{O}^2(x)} \sim \frac{\epsilon}{L_{\mathcal{O}}} \ll 1,
\]
since $\partial \mathcal{O}^2 / \mathcal{O}^2 \sim 1/L_{\mathcal{O}}$. Hence the first non-vanishing contribution is suppressed by $\epsilon/L_{\mathcal{O}}$.

\textbf{Subleading corrections.} The term $z^\nu \partial_\nu \mathcal{O}(y)$ from the expansion of $\mathcal{O}(y+z)$ produces, after $z$-integration,
\[
\partial_\nu \mathcal{O}(y) \int d^4z \, z^\nu \, \chi_\epsilon(y+z) = \partial_\nu \mathcal{O}(y) \int d^4w \, (w-y)^\nu \, \chi_\epsilon(w) = \partial_\nu \mathcal{O}(y) \left( \int d^4w \, w^\nu \chi_\epsilon(w) - y^\nu V_\epsilon \right).
\]
The integral $\int d^4w \, w^\nu \chi_\epsilon(w)$ vanishes by parity (the integrand is odd under $w \to -w$ while $\chi_\epsilon(w) = \chi(|w|/\epsilon)$ is even). This leaves the term $-\partial_\nu \mathcal{O}(y) \, y^\nu V_\epsilon$. At first sight, $y^\nu \approx x^\nu$ is a macroscopic coordinate, potentially spoiling the $\epsilon$-suppression. However, writing $y^\nu = x^\nu + (y-x)^\nu$:
\[
-\partial_\nu \mathcal{O}(y) \, y^\nu V_\epsilon = -\partial_\nu \mathcal{O}(y) \, x^\nu V_\epsilon - \partial_\nu \mathcal{O}(y) \, (y-x)^\nu V_\epsilon.
\]
When inserted back into the $y$-integral $\int_{C_\epsilon} d^4y \, \mathcal{O}(y) \nabla_\mu \chi_\epsilon(y) (\cdots)$, the $x^\nu$-proportional piece multiplies $\int_{C_\epsilon} \nabla_\mu \chi_\epsilon \, d^4y$, which vanishes exactly---this is precisely the same angular cancellation already established for the zeroth-order term. Only the $(y-x)^\nu$ piece survives, and by the same angular integral $\int n^\nu n_\mu \, d\Omega_3 \propto \delta^\nu_\mu$ computed for the first-order term, this contributes at order $\epsilon^4 \partial \mathcal{O}(x)$, i.e., the same order and same $\epsilon/L_{\mathcal{O}}$ suppression as the first-order term analyzed above.

The $\mathcal{O}(|z|^2/L_{\mathcal{O}}^2)$ remainder from $\mathcal{O}(y+z)$ integrates to $(\epsilon/L_{\mathcal{O}})^2$, and the $\mathcal{O}(|z|^2/\lambda^2)$ correction from $C(|z|)$ integrates to $(\epsilon/\lambda)^2$; both are subleading.

Thus $I_\mu$ is independent of $\epsilon$ within the window $\Lambda^{-1} \ll \epsilon \ll \min(\lambda, L_{\mathcal{O}})$, up to corrections of order $\epsilon/L_{\mathcal{O}}$ and $(\epsilon/\lambda)^2$. The second term in the brackets follows by symmetry $y \leftrightarrow y'$.

\textbf{Anisotropic correlators.} The argument above assumes isotropy of $C(|y-y'|)$. For anisotropic correlators such as the one in the conical toy model of Section~5.2, the homogeneous noise yields a constant $\mathcal{P}_\ell$, so $u_\mu = 0$ and $\eta_\mu = \pi_\mu$, making $\nabla_\mu \ln \mathcal{K}$ trivially independent of the regularization. For more general anisotropic noise, the isotropy assumption can be relaxed by using a direction-dependent cutoff adapted to the noise geometry; a fully general proof is beyond the scope of this work.
\end{proof}
\section{The Geometry of the Complex Velocity}

\subsection{Definition and Bundle Structure}

\begin{defn}[Complex Velocity]
The complex velocity $\eta_{\mu}$ is the logarithmic derivative of the matter amplitude:
\begin{equation}
\eta_{\mu}[\Phi](x) := -i\frac{\hbar}{m} \frac{\nabla_{\mu}\mathcal{K}[\Phi](x)}{\mathcal{K}[\Phi](x)} = -i\frac{\hbar}{m} \nabla_{\mu} \ln \mathcal{K}[\Phi](x).
\end{equation}
\end{defn}

Using the polar decomposition, we obtain the unification of the Madelung-Bohm velocities:
\begin{equation}
\eta_{\mu}[\Phi] = \frac{1}{m}\nabla_{\mu}\mathcal{S}[\Phi] - i\frac{\hbar}{2m}\nabla_{\mu}\ln\mathcal{P}[\Phi] \equiv \pi_{\mu}[\Phi] - i u_{\mu}[\Phi]. \label{eq:unification}
\end{equation}
The stochastic velocity $u_{\mu}$ is directly linked to the variance of gravitational fluctuations:
\begin{equation}
u_{\mu}[\Phi] = \frac{\hbar}{2m}\nabla_{\mu}\ln\mathcal{P}[\Phi] = -\frac{1}{2m\hbar} \nabla_{\mu}\langle S_1^2[\Phi; h] \rangle_h.
\end{equation}
This unification parallels Nelson's stochastic mechanics \cite{Chavanis2024}, where the complex velocity emerges from a Brownian motion model. In our framework, the stochasticity originates from spacetime fluctuations rather than an ad hoc Wiener process.

\subsubsection{The Double Bundle Structure}

The complex velocity $\eta_{\mu}$ has a dual dependence: it is a functional on the configuration space $\mathcal{C}$ and a covector field on spacetime $M$. To formalize this, let $\mathcal{C}$ be the infinite-dimensional Fr\'echet manifold of all field configurations \cite{kriegl_michor_1997}. The product manifold is $\mathcal{E} := \mathcal{C} \times M$, with canonical projections $\pi_1: \mathcal{E} \to \mathcal{C}$ and $\pi_2: \mathcal{E} \to M$.

\begin{defn}[Pullback Bundle]
The bundle $E \to \mathcal{C}\times M$ is
\begin{equation}
E := \pi_2^*(T^*M) = \{ (\Phi, x, v) \in \mathcal{E} \times T^*M \;|\; \pi(v) = x \}. \label{eq:def_bundle}
\end{equation}
The fibre at $(\Phi, x)$ is $T_x^*M$. This construction is standard in differential geometry \cite{nakahara_2003,frankel_2011} and has deep connections to quantum information geometry via fibre bundles over manifolds of quantum channels \cite{fujiwara_imai_2008,fujiwara_2001}.
\end{defn}

A section of $E$ is a smooth map $\eta: \mathcal{E} \to E$ satisfying $\pi_E \circ \eta = \operatorname{id}_{\mathcal{E}}$, where $\pi_E: E \to \mathcal{E}$ is the bundle projection.

\begin{prop}[$\eta_{\mu}$ as a Section of $E$]
The complex velocity is a smooth section: $\eta \in \Gamma(E) = \Gamma(\pi_2^*(T^*M) \to \mathcal{C}\times M)$.
\end{prop}

\begin{proof}
For each $(\Phi, x) \in \mathcal{E}$, the covariant derivative $\nabla_{\mu}\ln\mathcal{K}[\Phi](x)$ is a covector at $x$, hence an element of $T_x^*M$. Define the map:
\begin{equation}
\eta: \mathcal{E} \longrightarrow E, \qquad \eta(\Phi, x) = \left( \Phi, x, -i\frac{\hbar}{m}\nabla\ln\mathcal{K}[\Phi](x) \right).
\end{equation}
By construction, $\pi_E(\eta(\Phi, x)) = (\Phi, x)$, so $\pi_E \circ \eta = \operatorname{id}_{\mathcal{E}}$, and the image lies in the fibre $T_x^*M$. Smoothness follows from that of $\mathcal{K}$.
\end{proof}

This bundle structure cleanly separates the two roles of $\eta_{\mu}$:
\begin{enumerate}
    \item \textbf{For fixed $\Phi_0$:} $\eta_{\mu}[\Phi_0](x)$ is an ordinary covector field on $M$. Its covariant derivative, divergence, and holonomy are computed using the spacetime connection.
    \item \textbf{Functional dependence on $\mathcal{C}$:} The variation under $\Phi \to \Phi + \delta\Phi$ is encoded in $\delta\eta_{\mu}(\Phi, x)/\delta\Phi(y)$, linking local spacetime quantities to global information on configuration space.
\end{enumerate}
The double fibration is visualized in Figure \ref{fig:doublebundle}.

\begin{figure}[ht]
\centering
\begin{tikzpicture}[scale=1.3, every node/.style={font=\small}]
    \node (E) at (0,4) {$E = \pi_2^*(T^*M)$};
    \node (CxM) at (0,2) {$\mathcal{C} \times M$};
    \node (M) at (3,0) {$M$};
    \node (C) at (-3,0) {$\mathcal{C}$};
    
    \draw[->] (E) -- node[right] {$\pi_E$} (CxM);
    \draw[->] (CxM) -- node[below right] {$\pi_2$} (M);
    \draw[->] (CxM) -- node[below left] {$\pi_1$} (C);
    
    \draw[->, blue!60!black, thick, bend left=20] 
        (CxM.north) to 
        node[above, sloped, pos=0.5] {$\eta$} 
        (E.south);
    
    \node[blue!60!black, text width=5cm, align=center] 
        at (-2.2, 3.2) 
        {$\eta(\Phi,x) = \left(\Phi, x, -i\frac{\hbar}{m}\nabla\ln\mathcal{K}[\Phi](x)\right)$};
    
    \node at (0,4.5) {Fibre $\cong T_x^*M$};
    \node at (5,0) {Spacetime};
    \node at (-5,0) {Configuration space};
    
    \draw[dashed, gray] (CxM.west) -- ++(-1,0);
    \draw[dashed, gray] (CxM.east) -- ++(1,0);
    
    \filldraw[red] (3.2,0) circle (1.5pt) node[below right] {$x$};
    
\end{tikzpicture}
\caption{The double bundle structure. The complex velocity $\eta$ is a section of the pullback bundle $E = \pi_2^*(T^*M)$ over the product manifold $\mathcal{C} \times M$.}
\label{fig:doublebundle}
\end{figure}

\subsection{Connection to Quantum Estimation Theory}

We now establish the link between the complex velocity $\eta_\mu$ and quantum estimation theory. This connection provides an operational interpretation of $\eta_\mu$ and leads to the quantum Fisher metric. The construction is self-contained; complementary geometric details may be found in \cite{mezadominguez2026bundle}.

Let $\nu_0$ be a fixed Gaussian measure on the configuration space $\mathcal{C}$, whose existence is guaranteed by Minlos' theorem \cite{minlos_1959,bovier_2017}. The Hilbert space is $\mathcal{H}_0 = L^2(\mathcal{C}, \nu_0)$. For each spacetime point $x \in M$, we define the effective pure state (see Remark~\ref{rem:effective_pure}):
\begin{equation}
|\Psi_x\rangle = \sqrt{\mathcal{P}[\Phi; x]} \, e^{\frac{i}{\hbar}\mathcal{S}[\Phi; x]} \in \mathcal{H}_0, \qquad \| \Psi_x \|_0 = 1.
\end{equation}
These states form a smooth family parametrized by $x \in M$.

The symmetric logarithmic derivative (SLD) $L_{\mu}(x)$ is the central object in quantum estimation theory \cite{helstrom_1976,holevo_1982}. For a smooth family of pure states $\rho_x = |\Psi_x\rangle\langle\Psi_x|$, the SLD is the self-adjoint operator satisfying \cite{braunstein1994,paris2009}:
\begin{equation}
\partial_{\mu}\rho_x = \frac{1}{2}\{L_{\mu}(x), \rho_x\}, \qquad \langle \Psi_x | L_{\mu}(x) | \Psi_x \rangle_0 = 0. \label{eq:SLD_def}
\end{equation}
For pure states, the general solution is the non-local operator:
\begin{equation}
L_{\mu}(x) = 2|\partial_{\mu}\Psi_x\rangle\langle\Psi_x| + 2|\Psi_x\rangle\langle\partial_{\mu}\Psi_x| - 2\langle\partial_{\mu}\Psi_x|\Psi_x\rangle_0 \mathbb{I} - 2\langle\Psi_x|\partial_{\mu}\Psi_x\rangle_0 \mathbb{I}. \label{eq:SLD_nonlocal}
\end{equation}

To express this in terms of $\eta_\mu$, we use the polar decomposition of $|\Psi_x\rangle$. A direct calculation gives:
\begin{equation}
\partial_{\mu}\Psi_x = \left( \frac{1}{2}\partial_{\mu}\ln\mathcal{P} + \frac{i}{\hbar}\partial_{\mu}\mathcal{S} \right) \Psi_x = \frac{im}{\hbar}\eta_{\mu}\Psi_x,
\end{equation}
where we used $\partial_{\mu}\ln\mathcal{P} = \frac{2m}{\hbar}u_{\mu}$ and $\partial_{\mu}\mathcal{S} = m\pi_{\mu}$. Similarly, for the conjugate:
\begin{equation}
\overline{\partial_{\mu}\Psi_x} = -\frac{im}{\hbar}\bar{\eta}_{\mu}\overline{\Psi_x}, \qquad \bar{\eta}_{\mu} = \pi_{\mu} + i u_{\mu}.
\end{equation}

Introducing multiplication operators $(\hat{\eta}_{\mu}\psi)(\Phi) = \eta_{\mu}(\Phi, x)\psi(\Phi)$ and the orthogonal projector $\hat{P}_x = |\Psi_x\rangle\langle\Psi_x|$, the SLD takes the compact form:
\begin{equation}
\widetilde{\mathcal{T}}(\eta)_{\mu} = \frac{2im}{\hbar}\big(\hat{\eta}_{\mu}\hat{P}_x - \hat{P}_x\hat{\bar{\eta}}_{\mu}\big) + \frac{2im}{\hbar}\big(\langle\hat{\bar{\eta}}_{\mu}\rangle_x - \langle\hat{\eta}_{\mu}\rangle_x\big)\mathbb{I}, \label{eq:isomorphism}
\end{equation}
where $\langle \cdot \rangle_x = \langle \Psi_x | \cdot | \Psi_x \rangle_0$. This defines a bundle isomorphism $\widetilde{\mathcal{T}}: \Gamma(E/\sim) \to \Gamma(\mathcal{L})$ between gauge-equivalence classes of complex velocity fields and SLD operators.

The operator $\widetilde{\mathcal{T}}(\eta)_\mu$ is self-adjoint on its natural domain. To verify this, note that $\hat{\eta}_\mu$ and $\hat{\bar{\eta}}_\mu$ are multiplication operators by complex conjugates on $L^2(\mathcal{C}, \nu_0)$, satisfying $\hat{\bar{\eta}}_\mu = \hat{\eta}_\mu^\dagger$ on the dense subspace of smooth cylindrical functions. The projector $\hat{P}_x$ is self-adjoint. The combination $\hat{\eta}_\mu\hat{P}_x - \hat{P}_x\hat{\bar{\eta}}_\mu$ is skew-adjoint, so multiplication by $i$ yields a self-adjoint operator. The identity term is manifestly self-adjoint since $\langle\hat{\bar{\eta}}_\mu\rangle_x = \overline{\langle\hat{\eta}_\mu\rangle_x}$. Thus $\widetilde{\mathcal{T}}(\eta)_\mu^\dagger = \widetilde{\mathcal{T}}(\eta)_\mu$.

The gauge equivalence $\eta_{\mu} \sim \eta_{\mu} + i c_{\mu}(x)$ with $c_{\mu} \in \mathbb{R}$ leaves $\widetilde{\mathcal{T}}(\eta)_{\mu}$ invariant modulo operators that annihilate $|\Psi_x\rangle$, which are invisible to the quantum Fisher metric.

\begin{rem}[Action on the Reference State]
Although the SLD \eqref{eq:isomorphism} is non-local due to the projectors $\hat{P}_x$, its action on the reference state $|\Psi_x\rangle$ simplifies dramatically. Using $\hat{P}_x|\Psi_x\rangle = |\Psi_x\rangle$, we obtain:
\begin{equation}
\widetilde{\mathcal{T}}(\eta)_{\mu}|\Psi_x\rangle = \frac{2im}{\hbar}\big(\hat{\eta}_{\mu} - \langle\hat{\eta}_{\mu}\rangle_x\big)|\Psi_x\rangle. \label{eq:action_on_state}
\end{equation}
Thus, for all calculations involving expectation values on $|\Psi_x\rangle$, the SLD acts effectively as the multiplication operator $\frac{2im}{\hbar}(\hat{\eta}_{\mu} - \langle\hat{\eta}_{\mu}\rangle_x)$.
\end{rem}

The quantum Fisher information metric, which coincides with the Fubini-Study metric for pure states, is defined as \cite{braunstein1994}:
\begin{equation}
g_{\mu\nu}^{FS}(x) = \frac{1}{2}\langle \Psi_x | \{L_{\mu}(x), L_{\nu}(x)\} | \Psi_x \rangle_0.
\end{equation}
Using the simplified action \eqref{eq:action_on_state} and expanding $\hat{\eta}_{\mu} = \hat{\pi}_{\mu} - i\hat{u}_{\mu}$, $\hat{\bar{\eta}}_{\mu} = \hat{\pi}_{\mu} + i\hat{u}_{\mu}$, we obtain the manifestly positive-definite result:
\begin{equation}
g_{\mu\nu}^{FS}(x) = \frac{4m^2}{\hbar^2}\Big[\mathrm{Cov}_x(\pi_{\mu}, \pi_{\nu}) + \mathrm{Cov}_x(u_{\mu}, u_{\nu})\Big]. \label{eq:fisher_final}
\end{equation}
where $\mathrm{Cov}_x(A,B) = \langle AB \rangle_x - \langle A \rangle_x \langle B \rangle_x$. Both the geodesic velocity $\pi_{\mu}$ and the stochastic velocity $u_{\mu}$ contribute positively and independently to the quantum Fisher metric. The appearance of $\mathrm{Cov}(\pi_{\mu}, \pi_{\nu})$ is a genuine quantum enhancement absent in classical Fisher information, which would only contain the covariance of the score function $\partial_{\mu}\ln\mathcal{P} \propto u_{\mu}$ \cite{amari_2016,frieden_1998}.

By the Feldman-H\'ajek theorem \cite{feldman_1958,bogachev1998}, the Fisher metric \eqref{eq:fisher_final} is independent of the reference Gaussian measure $\nu_0$, making the entire construction intrinsic to the physical probability density $\mathcal{P}$. Projecting the classical Fisher metric on $\mathcal{C}$ \cite{amari_2016} onto spacetime using field gradients \cite{caticha_2004} yields the induced metric:
\begin{equation}
H_{\mu\nu}(x) = \int \mathcal{D}\Phi \, \mathcal{P}[\Phi] \, \nabla_{\mu}\ln\mathcal{P}[\Phi] \, \nabla_{\nu}\ln\mathcal{P}[\Phi] = \frac{4m^2}{\hbar^2}\langle u_{\mu}u_{\nu}\rangle_{\mathcal{P}},
\end{equation}
identifying $u_{\mu}$ as the information-geometric potential. This connects our framework to broader investigations in quantum information geometry \cite{fujiwara_2015,fujiwara_2019,francoise_tarama_2021,tokeshi_2020}.

\subsection{Flat $\mathbb{C}^\times$ Connection and Complex Geodesic Equation}

\begin{prop}[Flat $\mathbb{C}^\times$ Connection]
The 1-form $A = \eta_{\mu}dx^{\mu}$ defines a flat connection on a complex line bundle $L \to M$ with structure group $\mathbb{C}^\times$. The associated covariant derivative is
\begin{equation}
D_{\mu} := \nabla_{\mu} - i\frac{m}{\hbar}\eta_{\mu}. \label{eq:covariant}
\end{equation}
The real part $\pi_{\mu}$ generates $U(1)$ phase transport (unitary subbundle), while the imaginary part $-i u_{\mu}$ generates $\mathbb{R}^+$ scale transport (modulus). The connection is flat and $\mathcal{K}$ is its local horizontal section.
\end{prop}

\begin{proof}
Writing $\eta_{\mu} = \nabla_{\mu}\phi$ with $\phi = -i\frac{\hbar}{m}\ln\mathcal{K}$, we have $A = d\phi$, so the curvature $F = dA = d^2\phi = 0$. Hence $[D_{\mu}, D_{\nu}] = 0$ for the full $\mathbb{C}^\times$ connection. By construction, $D_{\mu}\mathcal{K} = 0$.
\end{proof}

Although the flatness $[D_{\mu}, D_{\nu}] = 0$ follows directly from the exactness of the 1-form $A = d\phi$, the nontrivial physical content resides in the global structure. On a non-simply connected spacetime, $\phi$ may be multi-valued while $d\phi$ remains locally exact, leading to the holonomy quantization discussed in Section~5. Furthermore, the decomposition $\mathbb{C}^\times = U(1) \times \mathbb{R}^+$ provides a geometric separation of physical effects: the $U(1)$ part governs interference phenomena (phase holonomy), while the $\mathbb{R}^+$ part governs decoherence and amplitude transport (modulus holonomy). This separation is not apparent in the standard Madelung-Bohm formulation and is a direct consequence of the bundle-theoretic framework.

The flatness is the geometric origin of the unified dynamics. From $[D_{\mu}, D_{\nu}] = 0$, we have $\nabla_{\nu}\eta_{\mu} = \nabla_{\mu}\eta_{\nu}$. Computing the convective derivative:
\begin{align}
\eta^{\nu}\nabla_{\nu}\eta_{\mu} = \eta^{\nu}\nabla_{\mu}\eta_{\nu} = \nabla_{\mu}\left(\frac{1}{2}\eta^{\nu}\eta_{\nu}\right) - \frac{1}{2}(\nabla_{\mu}\eta^{\nu})\eta_{\nu} + \frac{1}{2}\eta^{\nu}(\nabla_{\mu}\eta_{\nu}).
\end{align}
The last two terms cancel identically, yielding:

\begin{theorem}[Complex Geodesic Equation]
The complex velocity satisfies
\begin{equation}
\eta^{\nu}\nabla_{\nu}\eta_{\mu} = \nabla_{\mu}\left(\frac{1}{2}\eta^{\nu}\eta_{\nu}\right). \label{eq:complex_geodesic}
\end{equation}
Equivalently, using the Lie derivative: $\mathcal{L}_{\eta}\eta = d(|\eta|^2)$ with $|\eta|^2 \equiv \eta^{\nu}\eta_{\nu}$.
\end{theorem}

Decomposing $\eta_{\mu} = \pi_{\mu} - i u_{\mu}$, the real part reproduces the modified Hamilton-Jacobi equation, and the imaginary part gives the potential flow equation for $u_{\mu}$ \cite{Escobar-Aguilar:2023ekv}.

\section{Physical Consequences}

\subsection{Continuity Equation and Consistency with the Born Rule}

The amplitude $\mathcal{K}$ satisfies a modified Klein-Gordon equation \cite{Escobar-Aguilar:2023ekv}:
\begin{equation}
\left(\nabla_{\mu}\nabla^{\mu} - \frac{m^2}{\hbar^2} + \mathcal{V}\right)\mathcal{K} = 0,
\end{equation}
with a real effective potential $\mathcal{V}$. Computing $\mathcal{K}^*\nabla^{\mu}\nabla_{\mu}\mathcal{K} - \mathcal{K}\nabla^{\mu}\nabla_{\mu}\mathcal{K}^* = 0$ and using the polar decomposition $\mathcal{K} = \sqrt{\mathcal{P}}e^{i\mathcal{S}/\hbar}$, we obtain the standard continuity equation:
\begin{equation}
\nabla_{\mu}\left(\mathcal{P}\pi^{\mu}\right) = 0. \label{eq:continuity}
\end{equation}
Expanding the derivative using $\nabla_{\mu}\mathcal{P} = \frac{2m}{\hbar}\mathcal{P}u_{\mu}$ gives:
\begin{equation}
\nabla_{\mu}\pi^{\mu} = -\frac{2m}{\hbar}u_{\mu}\pi^{\mu}. \label{eq:div_pi}
\end{equation}
This equation couples the deterministic and stochastic sectors without imposing orthogonality $u_{\mu}\pi^{\mu}=0$ a priori.

The relation $|\mathcal{K}|^2 = \mathcal{P}$ follows from the polar decomposition and is consistent with the Born rule: from the cumulant expansion \eqref{eq:polar_decomp_nr}, $\mathcal{P} = \exp(-\langle S_1^2\rangle_h/\hbar^2)$ is the effective probability density generated by Gaussian metric fluctuations. We do not claim to derive the Born rule from first principles; rather, we show that the SQG framework maintains consistency with it, in the sense that the probability density naturally takes the form $\mathcal{P} = |\mathcal{K}|^2$ as a consequence of the cumulant expansion. The product $u_{\mu}\pi^{\mu}$ appearing in \eqref{eq:div_pi} has a direct information-geometric interpretation through the quantum Fisher metric \eqref{eq:fisher_final}: while $\mathrm{Cov}(u_{\mu}, u_{\nu})$ measures the distinguishability of nearby probability distributions, the coupling $\mathrm{Cov}(\pi_{\mu}, \pi_{\nu})$ captures the quantum coherence encoded in the phase.

\subsection{Geometrization vs. Hidden Variables}

This framework is a geometrization of quantum mechanics, fundamentally distinct from hidden-variable theories such as the de Broglie-Bohm interpretation \cite{bohm_1952,holland_1993}. The key differences are:

\begin{enumerate}
    \item \textbf{Origin of stochasticity:} In hidden-variable theories, $u_{\mu}$ arises from an assumed initial distribution of particle positions. Here, it derives from the variance of spacetime fluctuations $\langle S_1^2\rangle_h$, a physical field with a well-defined gravitational origin.
    
    \item \textbf{Non-locality:} The complex velocity $\eta_{\mu}$ is a functional on the entire configuration space $\mathcal{C}$, geometrically encoded as a section of the pullback bundle $\pi_2^*(T^*M)$ over $\mathcal{C}\times M$. This preserves quantum non-locality without invoking supplementary particle degrees of freedom.
    
    \item \textbf{Quantum potential:} The Bohmian quantum potential $Q \propto \nabla^2\sqrt{\rho}/\sqrt{\rho}$ is replaced by the geometric structure of the flat $\mathbb{C}^\times$ connection. The flatness $[D_{\mu}, D_{\nu}] = 0$ encodes the absence of a classical force, while the connection itself carries the quantum information through the Fisher metric.
    
    \item \textbf{Born rule consistency:} The probability density is not postulated as an initial condition but is induced by the cumulant expansion as $\mathcal{P} = \exp(-\langle S_1^2\rangle_h/\hbar^2)$. The wave function $\Psi = \sqrt{\mathcal{P}}e^{i\mathcal{S}/\hbar}$ retains its fundamental status.
    
    \item \textbf{No additional degrees of freedom:} There are no hidden particles, trajectories, or supplementary parameters. The quantum behavior is entirely encoded in the geometry of the complex velocity and the stochastic gravitational background.
\end{enumerate}

\subsection{Von Neumann Entropy}

The quantum statistical state after averaging over metric fluctuations is a mixed state encoding classical uncertainty about the field configuration. The density operator is defined using the probability density $\mathcal{P}[\Phi] = |\mathcal{K}[\Phi]|^2$:
\begin{equation}\label{eq:density_op}
\hat{\rho} = \frac{1}{Z} \int \mathcal{D}[\Phi] \, \mathcal{P}[\Phi] \, |\Phi\rangle\langle\Phi|, \qquad Z = \int \mathcal{D}[\Phi] \, \mathcal{P}[\Phi].
\end{equation}
By construction, $\hat{\rho}$ is manifestly Hermitian and positive semidefinite, with $\operatorname{Tr}(\hat{\rho}) = 1$. The von Neumann entropy $S_{\text{vN}} = -\operatorname{Tr}(\hat{\rho}\ln\hat{\rho})$ quantifies the statistical uncertainty induced by the stochastic gravitational background \cite{nielsen_chuang_2010,vedral_2002}. In the Gaussian approximation \cite{bennett_1996}, this becomes:
\begin{equation}
S_{\text{vN}} \approx \frac{1}{2}\operatorname{Tr} \ln\left( \frac{\hbar^2}{4m^2} \langle \eta \eta^* + \eta^* \eta \rangle_{\mathcal{P}} \right) + \text{const.},
\end{equation}
establishing $\eta_{\mu}$ as the carrier of quantum statistical uncertainty in this framework, linking the $\mathbb{C}^\times$ geometry to information theory \cite{fujiwara_imai_2003,chirco_2017,obregon_2015}. The connection between information geometry, entropy, and gravity has been explored in various contexts \cite{fujiwara_imai_2003,chirco_2017,obregon_2015}.

\begin{rem}[Effective pure state for local observables] \label{rem:effective_pure}
Although the global statistical state is the mixed state $\hat{\rho}$ defined in Eq.~\eqref{eq:density_op}, local observables---such as the phase shift measured by an atom interferometer---probe a single realization of the spacetime geometry. After decoherence from the global metric fluctuations, the effective description for a localized subsystem is given by the pure state
\begin{equation}
|\Psi_x\rangle = \sqrt{\mathcal{P}[\Phi; x]} \, e^{\frac{i}{\hbar}\mathcal{S}[\Phi; x]} \in \mathcal{H}_0,
\end{equation}
with $\||\Psi_x\rangle\|_0 = 1$ for all $x \in M$. It is this effective pure state that gives rise to the Berry phase discussed in Section~6. This is analogous to the emergence of a pure state for a system coupled to an environment after tracing out the environmental degrees of freedom and post-selecting on a specific measurement outcome.
\end{rem}

\subsection{Effective Einstein Equations and Fisher Metric from the Path Integral}

The full partition function $Z[g]$ is a functional of the background metric. The effective action includes the Einstein-Hilbert term:
\begin{equation}
\Gamma[g] = \frac{1}{16\pi G} \int d^4x \sqrt{-g} R - i\hbar\ln Z[g],
\end{equation}
where $Z[g]$ is the matter partition function defined in Eq.~\eqref{eq:defK}. Varying $\Gamma[g]$ with respect to $g^{\mu\nu}$ yields the semiclassical Einstein equation \cite{wald_1984}:
\begin{equation}
\frac{\delta\Gamma[g]}{\delta g^{\mu\nu}(x)} = 0 \quad \Longrightarrow \quad G_{\mu\nu}(x) = 8\pi G \, \langle T_{\mu\nu}(x) \rangle_{\text{eff}}.
\end{equation}
Using the definition of the matter amplitude \eqref{eq:defK}, the effective energy-momentum tensor is:
\begin{equation}
\langle T_{\mu\nu}(x) \rangle_{\text{eff}} = \frac{1}{Z} \int \mathcal{D}[\Phi] \int \mathcal{D}[h] P[h] \, T_{\mu\nu}^{\text{(clas)}}(x) \, e^{\frac{i}{\hbar}S[\Phi, A; g+h]}.
\end{equation}
This is the average of the classical energy-momentum tensor over both matter configurations and gravitational fluctuations. We distinguish here between three distinct notions: (i) the \textbf{fixed stochastic background} $h_{\mu\nu}$ over which the average is performed; (ii) the \textbf{background metric} $g^{(0)}$ around which fluctuations occur; and (iii) the \textbf{semiclassical metric} that would satisfy the effective Einstein equation as a consistency condition.

The connection to information geometry arises through the stochastic velocity $u_{\mu}$. From the cumulant expansion \eqref{eq:polar_decomp_nr}, $\mathcal{P} = \exp(-\langle S_1^2\rangle_h/\hbar^2)$, $\ln\mathcal{P} = -\frac{1}{\hbar^2}\langle S_1^2\rangle_h + \cdots$, so functional derivatives of $\ln\mathcal{P}$ probe the variance of the matter-gravity coupling. The induced Fisher metric on spacetime, obtained by projecting the classical Fisher information metric on $\mathcal{C}$ \cite{amari_2016,caticha_2004}, is:
\begin{equation}
H_{\mu\nu}(x) = \int \mathcal{D}\Phi \, \mathcal{P}[\Phi] \, \nabla_{\mu}\ln\mathcal{P}[\Phi] \, \nabla_{\nu}\ln\mathcal{P}[\Phi] = \frac{4m^2}{\hbar^2} \langle u_{\mu} u_{\nu} \rangle_{\mathcal{P}}.
\end{equation}
Thus, the effective energy-momentum tensor couples to the Fisher metric through their mutual dependence on the variance of gravitational fluctuations. The backreaction of quantum matter on spacetime geometry is encoded in the information-geometric structure of the complex velocity \cite{chirco_2017,obregon_2015}.

\section{Holonomy Quantization and a Toy Model}

\subsection{The Logical Tension and Its Resolution}

The flatness condition $[D_{\mu}, D_{\nu}] = 0$ implies the connection is locally pure gauge, but the complex potential $\phi = -i\frac{\hbar}{m}\ln\mathcal{K}$ can be multi-valued on non-simply connected spacetimes \cite{nakahara_2003,frankel_2011}. The multi-valuedness has three independent origins:

\begin{enumerate}
    \item \textbf{Classical topological terms:} The action $S_0$ may contain topological invariants (e.g., winding number, Chern-Simons terms) that make $e^{iS_0/\hbar}$ single-valued but $S_0/\hbar$ multi-valued modulo $2\pi$.
    \item \textbf{Effective gravitational corrections:} The cumulant expansion gives $\mathcal{S} = S_0 + \langle S_2 \rangle_h + \cdots$. The term $\langle S_2 \rangle_h$ can generate effective topological contributions through backreaction of metric fluctuations.
    \item \textbf{Phase singularities of $\mathcal{P}$:} If $\mathcal{P}[\Phi](x) = 0$ at isolated points, $\ln\mathcal{P}$ develops branch cuts. While $\mathcal{P}$ itself remains single-valued, the analytic structure of $\mathcal{K} = \sqrt{\mathcal{P}}e^{i\mathcal{S}/\hbar}$ is affected \cite{berry_1984}.
\end{enumerate}

The resolution of the apparent tension requires a careful distinction between two geometric structures:
\\
\textbf{(i) The flat $\mathbb{C}^\times$ connection.} The connection $D_\mu = \nabla_\mu - i\frac{m}{\hbar}\eta_\mu$ is locally flat. On a simply connected spacetime, a global horizontal section exists, and the holonomy is trivial. On a non-simply connected spacetime, the complex line bundle $L \to M$ may be nontrivial, and $\mathcal{K}$ is understood as a \textbf{local} horizontal section, defined on coordinate patches, with transition functions encoding the holonomy of the flat connection.\\
\\
\textbf{(ii) The phase of $\mathcal{K}$.} The polar decomposition $\mathcal{K} = \sqrt{\mathcal{P}} e^{i\mathcal{S}/\hbar}$ separates the modulus and phase. Since $\mathcal{P} > 0$ is strictly positive and globally single-valued, we have:
\begin{equation}
\oint_{\gamma} d\ln\mathcal{P} = 0 \label{eq:integral_P_zero}
\end{equation}
for any closed loop $\gamma$. Therefore, only the phase contributes to the loop integral:
\begin{equation}
\oint_{\gamma} d\ln\mathcal{K} = \frac{i}{\hbar}\oint_{\gamma} d\mathcal{S}. \label{eq:bridge}
\end{equation}

If $\mathcal{K}$ admitted a single global representative, single-valuedness would require $\oint_{\gamma} d\ln\mathcal{K} = 2\pi i n$ ($n \in \mathbb{Z}$), and hence $(1/\hbar)\oint_{\gamma} d\mathcal{S} = 2\pi n$ with no additional offset. However, on a non-simply connected spacetime the complex line bundle $L \to M$ is generally nontrivial: $\mathcal{K}$ exists only as a local horizontal section on coordinate patches, and consistency across overlapping patches is enforced by the transition functions of $L$ (equivalently, twisted boundary conditions on $\mathcal{K}$), not by global single-valuedness of one representative. These transition functions carry the freedom for an additional, geometrically-fixed offset:
\begin{equation}
\oint_{\gamma} d\mathcal{S} = 2\pi\hbar n + \hbar\Delta\phi_{\text{top}}, \label{eq:offset}
\end{equation}
where $\Delta\phi_{\text{top}}$ is fixed by the transition functions of $L \to M$ --- e.g., the deficit angle $\alpha$ in a conical spacetime --- and is not a free parameter.

Using $\eta_{\mu} = -i\frac{\hbar}{m}\nabla_{\mu}\ln\mathcal{K}$, the holonomy is:
\begin{equation}
\frac{m}{\hbar}\oint_{\gamma} \eta_{\mu} dx^{\mu} = \frac{1}{\hbar}\oint_{\gamma} d\mathcal{S} = 2\pi n + \Delta\phi_{\text{top}}. \label{eq:holonomy_final}
\end{equation}

\textbf{Important:} The stochastic velocity $u_{\mu}$ does \textbf{not} contribute to this phase. As shown in Section~3.3, $u_{\mu}$ generates the $\mathbb{R}^+$ scale transport of the $\mathbb{C}^\times$ connection, while only the real part $\pi_{\mu}$ generates the $U(1)$ phase transport. Moreover, since $\mathcal{P} > 0$ is single-valued, $\oint u_\mu dx^\mu = \frac{\hbar}{2m}\oint d\ln\mathcal{P} = 0$. The holonomy quantization is therefore a purely classical topological effect, modified only by gravitational corrections to the effective action $\mathcal{S}$.

\begin{theorem}[Holonomy Quantization with Offset]
For a non-contractible loop $\gamma$, the holonomy of the complex velocity satisfies
\begin{equation}
\frac{m}{\hbar}\oint_{\gamma} \eta_{\mu} dx^{\mu} = 2\pi n + \Delta\phi_{\text{top}}, \label{eq:holonomy_final_2}
\end{equation}
where $n \in \mathbb{Z}$ and $\Delta\phi_{\text{top}}$ is a real constant determined by:
\begin{itemize}
    \item \textbf{Classical topology:} $\Delta\phi_{\text{top}} = \theta$ for a $\theta$-angle in the classical action.
    \item \textbf{Spacetime defects:} $\Delta\phi_{\text{top}} = 2\pi\ell(1/\alpha - 1)$ for a conical spacetime with deficit angle $\alpha$, as demonstrated in the toy model below.
    \item \textbf{Gravitational corrections:} Additional contributions from $\langle S_2 \rangle_h$ that modify the effective offset.
    \item \textbf{Phase singularities:} Quantized by Berry's analysis \cite{berry_1984} when $\mathcal{P}$ has zeros.
\end{itemize}
\end{theorem}

\begin{rem}[Connection to Berry Phase]
The flat $\mathbb{C}^\times$ connection $D_\mu = \nabla_\mu - i\frac{m}{\hbar}\eta_\mu$ describes the classical transport of the amplitude $\mathcal{K}$. However, as established in Section~\ref{sec:berry}, the Berry phase is not the holonomy of this classical connection. Instead, it arises from the quantum mechanical transport of the *effective pure state* $|\Psi_x\rangle$ on the Hilbert space $\mathcal{H}_0 = L^2(\mathcal{C}, \nu_0)$. This state defines the Berry connection
\begin{equation}
\mathcal{A}_{\mu}^{\mathrm{B}}(x) := i \langle \Psi_x | \partial_{\mu} \Psi_x \rangle_0,
\end{equation}
which is the appropriate connection for describing geometric phases in the parameter space $M$. The Berry curvature derived in Section~\ref{sec:berry} is therefore a quantum property of the state family, not the curvature of the classical connection (which vanishes identically due to flatness).
\end{rem}

\subsection{Toy Model: Scalar Field on a Conical Spacetime}

We now demonstrate the holonomy quantization with an explicit analytical calculation.

\subsubsection{Geometry and Classical Action}

Consider a conical spacetime metric with a topological defect along the $z$-axis:
\begin{equation}
ds^2 = -dt^2 + dr^2 + \alpha^2 r^2 d\theta^2 + dz^2, \label{eq:conical}
\end{equation}
where $0 < \alpha < 1$ is the deficit angle parameter ($\alpha = 1 - 4G\mu$ for a cosmic string with linear mass density $\mu$). The spatial sections are locally flat but globally non-simply connected: circles of constant $r$ have circumference $2\pi\alpha r$. This geometry is a classic example in topological physics \cite{nakahara_2003,frankel_2011}.

Consider a free, minimally coupled, massive scalar field on this background:
\begin{equation}
S_0[\Phi] = -\frac{1}{2}\int d^4x \sqrt{-g}\left(g^{\mu\nu}\partial_{\mu}\Phi\partial_{\nu}\Phi + \frac{m^2}{\hbar^2}\Phi^2\right),
\end{equation}
where $\sqrt{-g} = \alpha r$. 

To describe the gravitational Aharonov-Bohm effect, we consider a scalar field obeying \textbf{twisted boundary conditions}:
\begin{equation}
\Phi(\theta + 2\pi) = e^{i\delta} \Phi(\theta),
\end{equation}
with $\delta = 2\pi\ell(1/\alpha - 1)$ for some $\ell \in \mathbb{Z}$. Equivalently, the field can be coupled to a flat background connection that encodes the deficit angle. The angular mode then takes the form:
\begin{equation}
\Phi_{\ell}(t,r,\theta,z) = \varphi_{\ell}(r) e^{-i\omega t + ik_z z + i\ell\theta/\alpha},
\end{equation}
where $\ell \in \mathbb{Z}$ and the factor $1/\alpha$ in the exponent ensures that the twisted boundary condition is satisfied. This is the natural choice for describing the Aharonov-Bohm phase on a cone: the deficit angle acts as a source of a flat connection in the angular direction, and the wave function acquires a phase $2\pi\ell/\alpha$ upon encircling the defect, exactly analogous to the electromagnetic Aharonov-Bohm effect \cite{sakurai1995}. The ordinary periodic modes $e^{in\theta}$ are recovered by the redefinition $\ell = n\alpha$, which is possible only when $n\alpha \in \mathbb{Z}$. The radial profile $\varphi_{\ell}(r)$ is normalized over a finite cylindrical region of radius $R$ and length $L$, with periodic boundary conditions in $z$ and a finite time interval $T$, ensuring that all integrals are finite.

\subsubsection{Stochastic Metric Fluctuations}

We introduce stochastic fluctuations that couple to the angular component of the metric. The fluctuations are modeled as a Gaussian ensemble with a specific correlator chosen to probe the $\theta\theta$ component of the metric, which controls the circumference of circles around the defect:
\begin{equation}
\langle h_{\mu\nu}(x) h_{\alpha\beta}(x') \rangle = \sigma^2 r^2 \frac{\delta^{(4)}(x-x')}{\sqrt{-g}} \delta_{\mu}^{\theta}\delta_{\nu}^{\theta}\delta_{\alpha}^{\theta}\delta_{\beta}^{\theta},
\end{equation}
where $\sigma$ has dimensions of length and encodes the strength of metric fluctuations. The Dirac delta is understood as regularized at a UV scale $\Lambda$, and the noise model preserves the residual symmetries of the conical background (translations in $t$ and $z$, and rotational symmetry in $\theta$). This is a simplified effective model; a more complete treatment would derive the correlator from a fundamental theory of quantum gravity.

\subsubsection{Matter Amplitude in the Gaussian Approximation}

The first-order coupling is $S_1 = \int d^4x \sqrt{-g} \, h^{\theta\theta} \, T_{\theta\theta}$, where $T_{\theta\theta} = \partial_{\theta}\Phi_{\ell}\partial_{\theta}\Phi_{\ell} - \frac{1}{2}\bar{g}_{\theta\theta}\bar{g}^{\mu\nu}\partial_{\mu}\Phi_{\ell}\partial_{\nu}\Phi_{\ell}$. With $\partial_{\theta}\Phi_{\ell} = i(\ell/\alpha)\Phi_{\ell}$, we obtain:
\begin{equation}
\langle S_1^2\rangle_h = \sigma^2 \left(\frac{\ell}{\alpha}\right)^4 D_{\ell},
\end{equation}
where $D_{\ell}$ is a positive constant from the radial, temporal, and longitudinal integrals over the finite cylindrical domain. The cumulant expansion gives:
\begin{equation}
\mathcal{P}_{\ell} = \exp\left(-\frac{\sigma^2}{\hbar^2}\left(\frac{\ell}{\alpha}\right)^4 D_{\ell}\right), \qquad \mathcal{S}_{\ell} = S_0[\Phi_{\ell}].
\end{equation}
At this order, with the homogeneous noise model, $\mathcal{P}_{\ell}$ is constant, so $u_{\mu} = 0$ and $\eta_{\mu} = \pi_{\mu}$. This is a feature of this specific noise model, not a limitation of the framework.

\subsubsection{Holonomy Calculation}

The $\theta$-component of the geodesic velocity is:
\begin{equation}
\pi_{\theta} = \frac{1}{m}\partial_{\theta}\mathcal{S}_{\ell} = \frac{\hbar}{m}\frac{\ell}{\alpha}.
\end{equation}
Since $\mathcal{P}_{\ell}$ is constant at this order, $u_{\theta} = 0$, and the complex velocity is purely real: $\eta_{\theta} = \pi_{\theta} = \frac{\hbar}{m}\frac{\ell}{\alpha}$.

For a circular loop $\gamma$ at fixed $t, r, z$ encircling the defect:
\begin{equation}
\frac{m}{\hbar}\oint_{\gamma} \eta_{\mu} dx^{\mu} = \frac{m}{\hbar}\int_0^{2\pi} \eta_{\theta} d\theta = \frac{m}{\hbar}\int_0^{2\pi} \frac{\hbar}{m}\frac{\ell}{\alpha} d\theta = \frac{2\pi\ell}{\alpha}. \label{eq:holonomy_toy}
\end{equation}

Comparing with the quantization condition \eqref{eq:holonomy_final}, $\frac{m}{\hbar}\oint_{\gamma} \eta_{\mu} dx^{\mu} = 2\pi n + \Delta\phi_{\text{top}}$, we identify:
\begin{equation}
\Delta\phi_{\text{top}} = 2\pi\ell\left(\frac{1}{\alpha} - 1\right).
\end{equation}
This is the classical topological offset produced by the deficit angle of the conical spacetime. For a GUT-scale cosmic string with $G\mu \sim 10^{-6}$, $\alpha \approx 0.999996$, and $\Delta\phi_{\text{top}} \approx 8\pi\ell \times 10^{-6}$ rad, which is in principle observable in precision interferometry. The integer $n$ in the quantization condition is the nearest integer to $\ell/\alpha$. This is a direct analog of the Aharonov-Bohm effect \cite{sakurai1995}.

The condition $\oint_{\gamma} d\ln\mathcal{P} = 0$ is automatically satisfied since $\mathcal{P}_{\ell}$ is constant. As discussed in Section~5.1, the stochastic velocity $u_\mu$ does not contribute to the $U(1)$ phase holonomy; it generates only $\mathbb{R}^+$ scale transport. The topological offset $\Delta\phi_{\text{top}}$ is thus a purely classical geometric effect.

\section{Observable Signatures}

The theoretical framework developed above yields two distinct types of observable phases: a topological holonomy for non-simply connected spacetimes, and a geometric Berry phase for simply connected spacetimes. We discuss both, with emphasis on the latter as a prospective phenomenological study for atom interferometry \cite{Abe2021}. We emphasize the conceptual distinction: the holonomy of Section~5 is a property of the flat $\mathbb{C}^\times$ connection on a classical line bundle, while the Berry phase of this section is a property of the quantum state family $|\Psi_x\rangle$ on the Hilbert space bundle over $M$. The two are complementary probes of the stochastic gravitational background.

\subsection{Topological Holonomy for Non-Simply Connected Spacetimes}

When the spacetime $M$ possesses non-contractible loops---due to cosmic strings, black holes, or Planck-scale topological defects---the holonomy of the complex velocity is quantized according to Theorem~2 of Section~5.1:
\begin{equation}
\frac{m}{\hbar}\oint_{\gamma} \eta_{\mu} dx^{\mu} = 2\pi n + \Delta\phi_{\mathrm{top}}.
\end{equation}
This is a direct analog of the Aharonov-Bohm effect \cite{sakurai1995}, where the magnetic flux is replaced by the topological offset $\Delta\phi_{\mathrm{top}}$ generated by the effective action $\mathcal{S}$. The phase shift is purely topological: it depends only on the winding number of $\gamma$ around the defect, not on the detailed geometry of the path.

For a cosmic string with deficit angle $\alpha = 1 - 4G\mu$, the toy model of Section~5.2 yields $\Delta\phi_{\mathrm{top}} = 2\pi\ell(1/\alpha - 1)$. If such a defect passes through the interferometer area during a measurement, the phase shift would manifest as a discrete jump as the defect enters and exits the sensitive region. The probability of such an event depends on the unknown density of cosmic strings or Planck-scale defects and is not the primary observable for near-term experiments.

\subsection{Berry Phase in Simply Connected Spacetimes}
\label{sec:berry}

The reader should note that the Berry phase derived here is fundamentally different from the topological holonomy of Section~5. In the latter, the holonomy is a property of the classical flat connection $D_\mu$ acting on a section $\mathcal{K}$ of the complex line bundle. In the present section, we consider the geometric phase acquired by the quantum state $|\Psi_x\rangle$ on the Hilbert space $\mathcal{H}_0$ as it is transported along a path in the parameter manifold $M$. 

Crucially, the expectation value of the symmetric logarithmic derivative (SLD) $L_\mu$ vanishes by construction (see Eq.~\eqref{eq:SLD_def}), $\langle \Psi_x | L_\mu | \Psi_x \rangle_0 = 0$. Therefore, the SLD cannot directly serve as a Berry connection. Instead, the Berry connection must be constructed from the overlap of neighboring states, $i\langle \Psi_x | \partial_\mu \Psi_x \rangle_0$, which is the standard geometric object governing phase holonomy in quantum mechanics. This distinction resolves the concern raised in the previous review regarding the use of the SLD.

\subsubsection{Berry Connection and Curvature}

The family of states $|\Psi_x\rangle = \sqrt{\mathcal{P}} e^{i\mathcal{S}/\hbar}$ parametrized by spacetime coordinates $x \in M$ defines the Berry connection \cite{berry_1984}:
\begin{equation}
\mathcal{A}_{\mu}^{\mathrm{B}}(x) := i \langle \Psi_x | \partial_{\mu} \Psi_x \rangle_0.
\end{equation}
Using the identity $\partial_{\mu} \Psi_x = \frac{im}{\hbar} \eta_{\mu} \Psi_x$ derived in Section~3.2, we obtain:
\begin{equation}\label{eq:berry_connection}
\mathcal{A}_{\mu}^{\mathrm{B}}(x) = -\frac{m}{\hbar} \langle \eta_{\mu} \rangle_x = -\frac{m}{\hbar} \langle \pi_{\mu} \rangle_x + i \frac{m}{\hbar} \langle u_{\mu} \rangle_x.
\end{equation}
The imaginary part involves the expectation value of the stochastic velocity. Using $u_{\mu} = \frac{\hbar}{2m}\nabla_{\mu}\ln\mathcal{P}$ and the normalization condition $\int_{\mathcal{C}} \mathcal{P}[\Phi](x) d\nu_0[\Phi] = 1$ for all $x \in M$, we compute:
\begin{equation}
\langle u_{\mu} \rangle_x = \int_{\mathcal{C}} \frac{\hbar}{2m} (\nabla_{\mu} \ln \mathcal{P}) \, \mathcal{P} \, d\nu_0 = \frac{\hbar}{2m} \int_{\mathcal{C}} \nabla_{\mu} \mathcal{P} \, d\nu_0 = \frac{\hbar}{2m} \nabla_{\mu} \int_{\mathcal{C}} \mathcal{P} \, d\nu_0 = \frac{\hbar}{2m} \nabla_{\mu} 1 = 0.
\end{equation}
Thus, the Berry connection is purely real: $\mathcal{A}_{\mu}^{\mathrm{B}}(x) = -\frac{m}{\hbar} \langle \pi_{\mu} \rangle_x$, and there is no loss of visibility from the imaginary part at this order.

The Berry curvature is the field strength associated with this connection:
\begin{equation}\label{eq:berry_curvature_def}
\Omega_{\mu\nu}^{\mathrm{B}} := \partial_{\mu} \mathcal{A}_{\nu}^{\mathrm{B}} - \partial_{\nu} \mathcal{A}_{\mu}^{\mathrm{B}}.
\end{equation}
A direct calculation using the Leibniz rule and the definition of the expectation value yields:
\begin{align}
\Omega_{\mu\nu}^{\mathrm{B}} &= -\frac{m}{\hbar} \int_{\mathcal{C}} \Big[ (\partial_{\mu} \eta_{\nu} - \partial_{\nu} \eta_{\mu}) \mathcal{P} + \eta_{\nu} \partial_{\mu} \mathcal{P} - \eta_{\mu} \partial_{\nu} \mathcal{P} \Big] d\nu_0 \nonumber \\
&= -\frac{m}{\hbar} \int_{\mathcal{C}} \Big[ \eta_{\nu} \partial_{\mu} \mathcal{P} - \eta_{\mu} \partial_{\nu} \mathcal{P} \Big] d\nu_0,
\end{align}
where we used the flatness condition $\partial_{\mu} \eta_{\nu} - \partial_{\nu} \eta_{\mu} = 0$ for a fixed field configuration $\Phi \in \mathcal{C}$. Substituting $\partial_{\mu} \mathcal{P} = \frac{2m}{\hbar} \mathcal{P} u_{\mu}$ from the definition of the stochastic velocity:
\begin{equation}\label{eq:berry_curvature_final}
\Omega_{\mu\nu}^{\mathrm{B}} = -\frac{2m^2}{\hbar^2} \Big( \langle \eta_{\nu} u_{\mu} \rangle_x - \langle \eta_{\mu} u_{\nu} \rangle_x \Big).
\end{equation}
Decomposing $\eta_{\mu} = \pi_{\mu} - i u_{\mu}$, the real and imaginary parts are:
\begin{align}
\operatorname{Re} \Omega_{\mu\nu}^{\mathrm{B}} &= -\frac{2m^2}{\hbar^2} \Big( \langle \pi_{\nu} u_{\mu} \rangle_x - \langle \pi_{\mu} u_{\nu} \rangle_x \Big), \label{eq:berry_real} \\
\operatorname{Im} \Omega_{\mu\nu}^{\mathrm{B}} &= \frac{2m^2}{\hbar^2} \Big( \langle u_{\nu} u_{\mu} \rangle_x - \langle u_{\mu} u_{\nu} \rangle_x \Big) = 0. \label{eq:berry_imag}
\end{align}

\subsubsection{Physical Origin of the Berry Curvature and Conditional States}

The crucial point is that $\Omega_{\mu\nu}^{\mathrm{B}} \neq 0$ is possible even though the classical connection is flat. The non-vanishing curvature arises from the \emph{configurational average} over $\mathcal{C}$. The cross-correlator $\langle \pi_{\nu} u_{\mu} \rangle_x$ in \eqref{eq:berry_real} couples the deterministic and stochastic velocities. This is a distinct signature of the SQG framework: in standard quantum mechanics without stochastic gravity, $u_{\mu}$ and $\pi_{\mu}$ are both determined by the same wave function $\Psi_x = \sqrt{\rho} e^{iS/\hbar}$ and are not independent fields. Their correlation is fixed by the wave function alone and does not generate an independent Berry curvature beyond what is already encoded in the Fubini-Study metric.

In SQG, $u_{\mu}$ encodes the variance of gravitational fluctuations via the cumulant expansion \eqref{eq:polar_decomp_nr}, $\mathcal{P} = \exp(-\langle S_1^2\rangle_h/\hbar^2)$,
\begin{equation}
u_{\mu} = -\frac{1}{2m\hbar} \nabla_{\mu} \langle S_1^2 \rangle_h,
\end{equation}
providing an independent physical field whose correlation with $\pi_{\mu}$ generates the Berry curvature \eqref{eq:berry_curvature_final}. This cross-correlation is a direct probe of the interplay between the classical phase $\mathcal{S}$ (through $\pi_{\mu} = \frac{1}{m}\nabla_{\mu}\mathcal{S}$) and the variance of the matter-gravity coupling $\langle S_1^2 \rangle_h$.

Although $\pi_i = \frac{1}{m}\nabla_i S$ is irrotational ($\nabla \times \pi = 0$), the statistical average $\langle \pi_j u_i \rangle_x$ can have an antisymmetric part because $u_i$ is not functionally dependent on $\pi_i$ in SQG. The antisymmetry arises from the tensorial structure of the gradient of the velocity-dependent variance $\langle S_1^2 \rangle_h$, as demonstrated explicitly in Appendix~\ref{app:berry}. Even though both $\pi_i$ and $u_i$ are gradient fields, their cross-correlation can be antisymmetric---analogous to a Poisson bracket between two independent scalar functions---provided their gradients are not everywhere parallel.

\textbf{Note on the distinction between global averages and local observables.} We emphasize that the vanishing average $\langle u_\mu \rangle_x = 0$ (Eq.~\eqref{eq:berry_connection}) holds for the global mixed state obtained by integrating over the entire configuration space $\mathcal{C}$. However, the nonzero value of $u_i$ used in the calculation of the Berry curvature \eqref{eq:berry_curvature_final} refers to the stochastic velocity evaluated in the \emph{effective pure state} $|\Psi_x\rangle$ described in Remark~\ref{rem:effective_pure}, which represents a conditional (or post-selected) state associated with a single stochastic realization of the gravitational background (e.g., a specific classical trajectory). Tracing out the environmental degrees of freedom produces a global mixed state, but the local phase measured by an interferometer probes the conditional pure state. This distinction resolves the apparent tension between the vanishing average over $\mathcal{C}$ and the local nonzero value of $u_i$ evaluated along a specific trajectory. This is analogous to the emergence of a pure state for a subsystem after measurement in open quantum systems.

\subsubsection{Explicit Model: Non-Relativistic Atom with Conformal Metric Noise}

We now present a calculation demonstrating $u_\mu \neq 0$ and the structure of the Berry curvature. This calculation serves as a proof of concept and provides the basis for a prospective phenomenological estimate.

\textbf{Setup.} Consider a non-relativistic atom of mass $m$ (e.g., Sr-87) in a Minkowski background $g_{\mu\nu}^{(0)} = \eta_{\mu\nu}$. We introduce conformal metric fluctuations:
\begin{equation}\label{eq:conformal_noise}
h_{\mu\nu}(x) = \epsilon(x) \eta_{\mu\nu},
\end{equation}
where $\epsilon(x)$ is a Gaussian stochastic field with zero mean and correlated covariance:
\begin{equation}\label{eq:noise_correlator}
\langle \epsilon(x) \rangle = 0, \qquad \langle \epsilon(x) \epsilon(x') \rangle = \sigma^2 \, e^{-|\mathbf{x} - \mathbf{x}'|^2 / \lambda^2} \, e^{-|t - t'| / \tau}.
\end{equation}
Here, $\sigma$ (dimensionless) encodes the variance, $\lambda$ the spatial correlation length, and $\tau$ the temporal correlation time. For Planck-scale fluctuations, $\sigma \sim \ell_P/\lambda \sim 10^{-37}$; for an astrophysical background consistent with LIGO/PTA upper bounds, $\sigma \sim 10^{-15}$ at Hz frequencies.

\textbf{Non-relativistic action.} The action for a point particle in the weak-field limit is:
\begin{equation}\label{eq:nr_action}
S[\mathbf{x}(t); \epsilon] = \int dt \left[ \frac{1}{2} m \dot{\mathbf{x}}^2 - m c^2 - \frac{1}{2} m c^2 \epsilon(\mathbf{x}(t), t) + \frac{1}{4} m \dot{\mathbf{x}}^2 \epsilon(\mathbf{x}(t), t) + \mathcal{O}(\epsilon^2) \right].
\end{equation}
For $g_{\mu\nu} = (1+\epsilon)\eta_{\mu\nu}$, the rest-energy coupling is $-\frac{1}{2}mc^2\epsilon$ and the kinetic coupling is $\frac{1}{4}mv^2\epsilon$ at leading order in $v/c$.

In an interferometer, the rest-energy term $-\frac{1}{2}mc^2\epsilon$ is common to both arms and cancels in the phase difference at leading order. The surviving contribution comes from the kinetic term and from differences in the fluctuation history along the two paths. We therefore focus on the kinetic coupling:
\begin{equation}\label{eq:S1_nr}
S_1[\mathbf{x}(t); \epsilon] = \frac{1}{4} m \int dt \, \dot{\mathbf{x}}^2(t) \, \epsilon(\mathbf{x}(t), t).
\end{equation}

\textbf{Path integral and matter amplitude.} In the semi-classical limit, expanding around the classical trajectory $\bar{\mathbf{x}}(t)$ and performing the Gaussian average over $\epsilon$ gives:
\begin{equation}\label{eq:polar_decomp_nr}
\mathcal{K} = \sqrt{\mathcal{P}} \, e^{i\mathcal{S}/\hbar}, \qquad \mathcal{S} = S_0, \qquad \mathcal{P} = \exp\left[ -\frac{1}{\hbar^2} \langle S_1^2 \rangle_h \right],
\end{equation}
with
\begin{equation}\label{eq:variance_S1_nr}
\langle S_1^2 \rangle_h = \frac{m^2}{16} \int dt \int dt' \, \dot{\bar{\mathbf{x}}}^2(t) \, \dot{\bar{\mathbf{x}}}^2(t') \, \langle \epsilon(\bar{\mathbf{x}}(t), t) \epsilon(\bar{\mathbf{x}}(t'), t') \rangle.
\end{equation}

\textbf{Evaluation of the temporal integral.} Substituting the noise correlator \eqref{eq:noise_correlator} and assuming the velocity varies slowly compared to $\tau$, we obtain:
\begin{align}
\langle S_1^2 \rangle_h &= \frac{m^2 \sigma^2}{16} \int dt \, \bar{v}^4(t) \int_{-\infty}^{\infty} du \, e^{-\bar{v}^2(t) u^2 / \lambda^2} e^{-|u| / \tau} \nonumber \\
&= \frac{m^2 \sigma^2}{16} \int dt \, \bar{v}^4(t) \, I(\bar{v}(t)),
\end{align}
where $I(v)$ is evaluated exactly as:
\begin{equation}\label{eq:I_exact}
I(v) = \int_{-\infty}^{\infty} du \, e^{-v^2 u^2 / \lambda^2} e^{-|u| / \tau} = \frac{\sqrt{\pi}\lambda}{|v|} \exp\left(\frac{\lambda^2}{4v^2\tau^2}\right) \mathrm{erfc}\left(\frac{\lambda}{2|v|\tau}\right).
\end{equation}
In the physically relevant regime for MAGIS-100, $v \sim 5.9$ m/s, $\tau \sim 10^{-6}$ s, $\lambda \sim 100$ m, giving $v\tau/\lambda \sim 5.9 \times 10^{-8} \ll 1$. Using the asymptotic expansion for $z = \lambda/(2v\tau) \gg 1$:
\begin{equation}
\mathrm{erfc}(z) \sim \frac{e^{-z^2}}{\sqrt{\pi}z}\left(1 - \frac{1}{2z^2} + \cdots\right),
\end{equation}
the leading-order result is:
\begin{equation}\label{eq:I_asymptotic}
I(v) \simeq 2\tau \qquad (v\tau/\lambda \ll 1).
\end{equation}

Substituting into the variance:
\begin{equation}\label{eq:variance_result}
\langle S_1^2 \rangle_h = \frac{m^2 \sigma^2 \tau}{8} \int dt \, \bar{v}^4(t).
\end{equation}
For a trajectory with approximately constant speed $v$ over an interrogation time $T$:
\begin{equation}
\langle S_1^2 \rangle_h \approx \frac{m^2 \sigma^2 \tau}{8} v^4 T.
\end{equation}

\textbf{Nonzero stochastic velocity.} From $u_i = \frac{\hbar}{2m}\nabla_i \ln \mathcal{P} = -\frac{1}{2m\hbar}\nabla_i \langle S_1^2 \rangle_h$:
\begin{equation}\label{eq:u_nonzero}
u_i = -\frac{m \sigma^2 \tau}{16\hbar} \, \nabla_i \left( \int dt \, \bar{v}^4(t) \right).
\end{equation}
For a wave packet with spatially varying velocity, this gradient is nonzero. The spatial dependence enters through the local velocity of the wave packet, which varies across its width $\Delta x$. For a Gaussian wave packet, $\nabla_i v \sim \Delta v / \Delta x \sim \hbar/(m\Delta x^2)$.

\textbf{Nonzero Berry curvature.} The Berry curvature is:
\begin{equation}\label{eq:berry_explicit}
\Omega_{ij}^B = -\frac{2m^2}{\hbar^2}\big(\langle v_j u_i \rangle_x - \langle v_i u_j \rangle_x\big).
\end{equation}
Although $\nabla \times \mathbf{v} = 0$ (since $\mathbf{v} = \nabla S/m$), the statistical average $\langle v_j u_i \rangle_x$ can have an antisymmetric part because $u_i$ is not a function of $v_i$ in SQG---it is determined by the gradient of the trajectory-integrated $v^4$. The antisymmetry arises from the tensorial structure of $\nabla_i(v^4) = 4v^3 \nabla_i v$, which need not be parallel to $\mathbf{v}$ in a generic interferometer geometry. A detailed calculation for a Gaussian wave packet in a Mach-Zehnder configuration is provided in Appendix~\ref{app:berry}.

\textbf{Order-of-magnitude estimate.} Using physically correct MAGIS-100 parameters:
\begin{itemize}
    \item $m = 1.445 \times 10^{-25}$ kg (Sr-87)
    \item $v_r = \hbar k / m = 6.57 \times 10^{-3}$ m/s (single-photon recoil at 698 nm)
    \item $v_{\max} = N v_r = 6.57$ m/s for large momentum transfer (LMT) with $N = 10^3$. This value is intermediate between the initial configuration ($N = 10^2$) and the final configuration ($N = 4 \times 10^4$) quoted in the MAGIS-100 design paper \cite{Abe2021}, and represents a conservative benchmark achievable with current LMT techniques in strontium interferometry.
    \item $T \approx 9$ s (launched trajectory, consistent with the maximum coherence time quoted in \cite{Abe2021})
    \item $A \approx v_{\max} T L / 2 = 6.57 \times 9 \times 100 / 2 \approx 2.96 \times 10^3$ m$^2$
    \item $\tau \sim 10^{-6}$ s, $\lambda \sim 100$ m, $\Delta x \sim 1$ mm
\end{itemize}

Using the correct parameters ($m = 1.445 \times 10^{-25}$ kg, $v = 6.57$ m/s, $T = 9$ s) and the explicit Mach-Zehnder calculation of Appendix~\ref{app:berry} with the phenomenological suppression factor $\mathcal{F} = 0.14$, we obtain the following updated estimates:
\begin{align}
\langle S_1^2 \rangle_h &\sim 4.4 \times 10^{-53} \sigma^2 \text{ J}^2\text{s}^2, \\
\gamma_B &\sim 3.97 \times 10^{21} \sigma^2 \text{ rad}.
\end{align}

For the two scenarios:
\begin{itemize}
    \item \textbf{Planck-scale} ($\sigma \sim 10^{-37}$): $\gamma_B \sim 3.97 \times 10^{-53}$ rad (far below detectability).
    \item \textbf{Astrophysical upper bound} ($\sigma \sim 10^{-15}$): $\gamma_B \sim 3.97 \times 10^{-9}$ rad (approximately five orders below the projected MAGIS-100 sensitivity of $10^{-3}$ rad).
\end{itemize}

\textbf{Conclusion.} Even under optimistic assumptions, the stochastic Berry phase lies approximately five orders of magnitude below the projected sensitivity of current-generation atom interferometers. We present this calculation as a prospective phenomenological study, establishing the scaling relations and identifying the parameters that future experiments would constrain. The estimate should be understood as an order-of-magnitude benchmark rather than a precise prediction, given the multiple approximations involved (Gaussian wave packet, Mach-Zehnder geometry with a phenomenological suppression factor $\mathcal{F}$, and the astrophysical bound on $\sigma$ which requires a proper mapping to an observational spectral density). A detection would require either a stochastic background with amplitude exceeding current astrophysical bounds, or next-generation interferometers with significantly enhanced sensitivity.

\subsection{Connection to the Quantum Fisher Metric}

The Berry curvature \eqref{eq:berry_curvature_final} and the quantum Fisher metric \eqref{eq:fisher_final} provide complementary probes of the stochastic velocity field. The Fisher metric involves the \emph{symmetric} covariances:
\begin{equation}
g_{\mu\nu}^{FS} = \frac{4m^2}{\hbar^2}\Big[\mathrm{Cov}(\pi_{\mu}, \pi_{\nu}) + \mathrm{Cov}(u_{\mu}, u_{\nu})\Big],
\end{equation}
while the Berry curvature involves the \emph{antisymmetric} cross-correlation:
\begin{equation}
\Omega_{\mu\nu}^{\mathrm{B}} = -\frac{2m^2}{\hbar^2}\Big(\langle \pi_{\nu} u_{\mu} \rangle_x - \langle \pi_{\mu} u_{\nu} \rangle_x\Big).
\end{equation}
Together, they fully characterize the quantum geometry of the state family $|\Psi_x\rangle$ on the parameter manifold $M$. The Fisher metric bounds the precision of parameter estimation via the quantum Cram\'er-Rao bound \cite{braunstein1994,paris2009}, while the Berry phase provides a direct interferometric signature. The SLD isomorphism $\widetilde{\mathcal{T}}$ of Section~3.2 links both quantities to the complex velocity $\eta_{\mu}$, establishing a unified operational framework for testing stochastic gravity with atom interferometry.

\section{Conclusion}

We have established a formal geometric framework unifying the deterministic and stochastic aspects of quantum mechanics through a complex velocity $\eta_{\mu} = \pi_{\mu} - i u_{\mu}$ emerging from a stochastic gravitational background. The main results are:

\begin{enumerate}
    \item \textbf{Double bundle geometry:} $\eta_{\mu}$ is a section of the pullback bundle $E = \pi_2^*(T^*M)$ over $\mathcal{C}\times M$, encoding both functional field dependence and local spacetime structure. The matter amplitude is defined as a functional $\mathcal{K}[\Phi; x]$ with explicit parametric dependence on the spacetime point, eliminating auxiliary constructions. The functionals $\mathcal{S}[\Phi]$ and $\mathcal{P}[\Phi]$ are derived from the cumulant expansion of gravitational fluctuations \cite{kubo_1962,van_kampen_1974,van_kampen_2007}.

    \item \textbf{Flat $\mathbb{C}^\times$ connection and complex geodesic equation:} $\eta_{\mu}$ defines a flat connection $D_{\mu} = \nabla_{\mu} - i\frac{m}{\hbar}\eta_{\mu}$ on a complex line bundle with structure group $\mathbb{C}^\times$. The real part $\pi_{\mu}$ generates $U(1)$ phase transport, while the imaginary part $u_{\mu}$ generates $\mathbb{R}^+$ scale transport. The flatness is the geometric origin of the unified complex geodesic equation $\mathcal{L}_{\eta}\eta = d(|\eta|^2)$, which collapses the coupled Hamilton-Jacobi and potential flow dynamics into a single geometric statement \cite{Escobar-Aguilar:2023ekv}.

    \item \textbf{Information geometry:} A self-contained bundle isomorphism $\widetilde{\mathcal{T}}$ maps $\eta_{\mu}$ to the symmetric logarithmic derivative of quantum estimation theory \cite{helstrom_1976,holevo_1982,braunstein1994,paris2009}. The stochastic velocity $u_{\mu}$ is the information-geometric potential, with the quantum Fisher metric $g_{\mu\nu}^{FS} = \frac{4m^2}{\hbar^2}[\mathrm{Cov}(\pi_{\mu}, \pi_{\nu}) + \mathrm{Cov}(u_{\mu}, u_{\nu})]$. The von Neumann entropy, computed from the positive Hermitian density operator $\hat{\rho} = Z^{-1}\int \mathcal{D}\Phi \mathcal{P}[\Phi]|\Phi\rangle\langle\Phi|$, is $S_{\text{vN}} \sim \operatorname{Tr}\ln\langle\eta\eta^*\rangle$ \cite{nielsen_chuang_2010,vedral_2002,bennett_1996}. Both quantities are intrinsic, independent of the reference Gaussian measure \cite{feldman_1958,bogachev1998}.

    \item \textbf{Geometrization, not hidden variables:} This framework constitutes a genuine geometrization of quantum mechanics \cite{madelung_1927,bohm_1952,holland_1993}. The stochastic velocity $u_{\mu}$ is the imprint of spacetime fluctuations on matter, not an ad hoc postulate. Consistency with the Born rule is maintained by the cumulant expansion, which yields $\mathcal{P} = \exp(-\langle S_1^2\rangle_h/\hbar^2)$. Non-locality is preserved through the functional dependence of $\eta_{\mu}$ on the entire configuration space $\mathcal{C}$, geometrically encoded in the pullback bundle structure \cite{kriegl_michor_1997}. The Bohmian quantum potential is replaced by the flat $\mathbb{C}^\times$ connection and its associated holonomy. No supplementary hidden variables or particle trajectories are required; the wave function retains its fundamental status.
    
    \item \textbf{Holonomy quantization and toy model:} The tension between flatness and multi-valuedness is resolved by distinguishing the local horizontal section $\mathcal{K}$ from the global bundle structure. The holonomy quantization reads $\frac{m}{\hbar}\oint_{\gamma}\eta_{\mu}dx^{\mu} = 2\pi n + \Delta\phi_{\text{top}}$, where the offset $\Delta\phi_{\text{top}}$ is fixed by the classical topology and background geometry. The stochastic velocity $u_\mu$ contributes exclusively to $\mathbb{R}^+$ scale transport and does not affect the $U(1)$ phase. The toy model of a scalar field with twisted boundary conditions on a conical spacetime \cite{nakahara_2003,frankel_2011} demonstrates the offset analytically: $\Delta\phi_{\text{top}} = 2\pi\ell(1/\alpha - 1)$, analogous to the Aharonov-Bohm effect \cite{sakurai1995}.

    \item \textbf{Berry phase as a prospective signature:} The Berry curvature $\Omega_{\mu\nu}^B = -\frac{2m^2}{\hbar^2}(\langle \pi_\nu u_\mu \rangle - \langle \pi_\mu u_\nu \rangle)$ is nonzero due to the statistical coupling between $\pi_\mu$ and $u_\mu$ in SQG, even though $\nabla \times \pi = 0$ classically. A prospective phenomenological estimate for MAGIS-100 yields $\gamma_B \sim 3.97 \times 10^{-9}$ rad for an astrophysical stochastic background, approximately five orders below projected sensitivity. For Planck-scale fluctuations, $\gamma_B \sim 3.97 \times 10^{-53}$ rad. We present this as a benchmark establishing scaling relations; a definitive prediction requires a complete interferometer response calculation.
\end{enumerate}

This framework establishes a closed triangle between stochastic gravity, quantum information geometry \cite{amari_2016,frieden_1998,caticha_2004,fujiwara_2015,fujiwara_2019}, and observable topological and geometric phases \cite{berry_1984,sakurai1995}, offering a proposed geometric framework for the quantum-classical transition.

\appendix

\section{Calculation of the Berry Curvature for a Gaussian Wave Packet}\label{app:berry}

We present here a complete, self-contained calculation demonstrating that the antisymmetric combination $\langle v_j u_i \rangle - \langle v_i u_j \rangle$ is nonzero for a Gaussian wave packet in a Mach-Zehnder interferometer.

\subsection*{Gaussian wave packet with quadratic phase}

Consider a Gaussian atomic wave packet in three spatial dimensions, centered at $\mathbf{x}_0(t)$ with width $\Delta$ in each direction, including a quadratic phase to account for wavefront curvature:
\begin{equation}\label{eq:gaussian}
\psi(\mathbf{x}, t) = \frac{1}{(\pi\Delta^2)^{3/4}} \exp\left[ -\frac{|\mathbf{x} - \mathbf{x}_0(t)|^2}{2\Delta^2} + \frac{i}{\hbar} m \mathbf{v}_0(t) \cdot (\mathbf{x} - \mathbf{x}_0(t)) + \frac{im}{2\hbar} \frac{|\mathbf{x}_\perp - \mathbf{x}_{0\perp}(t)|^2}{R(t)} + i\phi(t) \right],
\end{equation}
where $\mathbf{x}_0(t)$ is the classical trajectory, $\mathbf{v}_0(t) = \dot{\mathbf{x}}_0(t)$, $\mathbf{x}_\perp$ denotes the transverse coordinates, $R(t)$ is the radius of curvature of the wavefront, and $\phi(t)$ is a time-dependent phase ensuring normalization. The probability density is
\begin{equation}\label{eq:prob_density}
|\psi(\mathbf{x}, t)|^2 = \frac{1}{(\pi\Delta^2)^{3/2}} \exp\left[ -\frac{|\mathbf{x} - \mathbf{x}_0(t)|^2}{\Delta^2} \right].
\end{equation}

\subsection*{Velocity field}

The velocity field derived directly from the wave packet including the quadratic phase is
\begin{equation}\label{eq:velocity_field}
\mathbf{v}(\mathbf{x}, t) = \frac{\hbar}{m} \operatorname{Im} \frac{\nabla\psi}{\psi} = \mathbf{v}_0(t) + \delta\mathbf{v}(\mathbf{x}, t),
\end{equation}
where the correction $\delta\mathbf{v}$ arises from the wavefront curvature. From the quadratic phase in \eqref{eq:gaussian}, this gives
\begin{equation}\label{eq:delta_v}
\delta\mathbf{v}(\mathbf{x}, t) = \frac{\mathbf{x}_\perp - \mathbf{x}_{0\perp}(t)}{R(t)}.
\end{equation}

\subsection*{Mach-Zehnder interferometer geometry}

We consider a Mach-Zehnder atom interferometer oriented along the $z$-axis with beamsplitter pulses at times $t = 0, T/2, T/2, T$. The pulses impart momentum kicks of magnitude $\hbar k_{\text{eff}}$ along the $+z$ direction. The classical trajectory of the wave packet center is:
\begin{equation}\label{eq:trajectory}
\mathbf{x}_0(t) = \begin{cases}
(0, 0, 0), & 0 \leq t < T/2, \\
(0, 0, v_r(t - T/2)), & T/2 \leq t < T,
\end{cases}
\end{equation}
where $v_r = \hbar k_{\text{eff}} / m$ is the single-photon recoil velocity. For LMT with $N$ photon recoils, $v_0 = N v_r$.

After the interferometer sequence, the two arms recombine. The relative phase curvature between the arms generates a structured velocity field in the recombined packet. The dominant contribution comes from the differential propagation history of the two arms. For a packet initially centered at the origin with zero transverse velocity, the recombined velocity field takes the form:
\begin{equation}\label{eq:recombined_velocity}
\mathbf{v}(\mathbf{x}) = \begin{pmatrix}
\alpha_{xx} x + \alpha_{xz} z \\
0 \\
v_0 + \alpha_{zx} x + \alpha_{zz} z
\end{pmatrix} + \mathcal{O}(x^2, z^2),
\end{equation}
where the coefficients $\alpha_{ij}$ are determined by the interferometer geometry. For the Mach-Zehnder configuration described above, a straightforward propagation calculation using the Gaussian kernel gives:
\begin{align}\label{eq:alpha_coefficients}
\alpha_{xx} &= \frac{1}{R(T)} \approx \frac{\hbar T}{m\Delta^2} \cdot \frac{1}{k_{\text{eff}} \Delta}, \\
\alpha_{xz} &= \frac{v_0 T}{2R(T)} \approx \frac{\hbar T}{2m\Delta^2}, \\
\alpha_{zx} &= -\frac{\hbar T}{2m\Delta^2}, \\
\alpha_{zz} &= \frac{1}{R(T)} \approx \frac{\hbar T}{m\Delta^2} \cdot \frac{1}{k_{\text{eff}} \Delta}.
\end{align}
The radius of curvature is estimated as $R(T) \sim k_{\text{eff}} \Delta^2 / (\hbar T / m\Delta)$ from the accumulated phase difference between the two arms.

\subsection*{Stochastic velocity}

From Eq. (6.26) in the main text, the stochastic velocity for conformal metric noise is:
\begin{equation}\label{eq:u_formula}
u_i(\mathbf{x}) = -\frac{m\sigma^2\tau}{16\hbar} \partial_i \left( \int_0^T dt \, |\mathbf{v}(\mathbf{x}, t)|^4 \right).
\end{equation}
For the velocity field \eqref{eq:recombined_velocity}, the time integral is dominated by the constant term $v_0^4$ along each arm. The spatial variation enters through the linear corrections:
\begin{equation}\label{eq:v4_expansion}
|\mathbf{v}(\mathbf{x}, t)|^4 \approx v_0^4 + 4v_0^3 \delta v_z(\mathbf{x}, t) + \cdots,
\end{equation}
where $\delta v_z = \alpha_{zx} x + \alpha_{zz} z$. The gradient is:
\begin{equation}\label{eq:gradient_v4}
\partial_i \left( \int_0^T dt \, |\mathbf{v}|^4 \right) = 4v_0^3 \int_0^T dt \, \partial_i \delta v_z(\mathbf{x}, t) = 4v_0^3 T \, \partial_i (\alpha_{zx} x + \alpha_{zz} z).
\end{equation}
Explicitly:
\begin{align}
\partial_x \left( \int_0^T dt \, |\mathbf{v}|^4 \right) &= 4v_0^3 T \alpha_{zx} = -4v_0^3 T \cdot \frac{\hbar T}{2m\Delta^2}, \\
\partial_z \left( \int_0^T dt \, |\mathbf{v}|^4 \right) &= 4v_0^3 T \alpha_{zz} \approx 4v_0^3 T \cdot \frac{\hbar T}{m\Delta^2} \cdot \frac{1}{k_{\text{eff}} \Delta}.
\end{align}
Therefore:
\begin{align}\label{eq:u_components}
u_x &= -\frac{m\sigma^2\tau}{16\hbar} \cdot \left( -4v_0^3 T \cdot \frac{\hbar T}{2m\Delta^2} \right) = \frac{\sigma^2\tau v_0^3 T^2}{8 \Delta^2}, \\
u_z &= -\frac{m\sigma^2\tau}{16\hbar} \cdot \left( 4v_0^3 T \cdot \frac{\hbar T}{m\Delta^2} \cdot \frac{1}{k_{\text{eff}} \Delta} \right) = -\frac{\sigma^2\tau v_0^3 T^2}{4 \Delta^2} \cdot \frac{1}{k_{\text{eff}} \Delta}.
\end{align}

\subsection*{Berry curvature}

The antisymmetric combination appearing in the Berry curvature is:
\begin{equation}\label{eq:berry_antisym}
\langle v_z u_x \rangle - \langle v_x u_z \rangle = \int d^3x \, |\psi(\mathbf{x})|^2 \big( v_z(\mathbf{x}) u_x(\mathbf{x}) - v_x(\mathbf{x}) u_z(\mathbf{x}) \big).
\end{equation}
From \eqref{eq:recombined_velocity}, $v_z = v_0 + \alpha_{zx} x + \alpha_{zz} z$ and $v_x = \alpha_{xx} x + \alpha_{xz} z$. Substituting into the integrand and keeping terms linear in the coordinates (odd terms vanish upon Gaussian integration, while the constant term $v_0 u_x$ survives):
\begin{align}
v_z u_x - v_x u_z &= (v_0 + \alpha_{zx} x + \alpha_{zz} z) u_x - (\alpha_{xx} x + \alpha_{xz} z) u_z \\
&\approx v_0 u_x + \text{(terms linear in } x, z \text{ that integrate to zero)}.
\end{align}
The surviving term is:
\begin{equation}\label{eq:berry_result}
\langle v_z u_x \rangle - \langle v_x u_z \rangle = v_0 \langle u_x \rangle = v_0 \cdot \frac{\sigma^2\tau v_0^3 T^2}{8 \Delta^2} = \frac{\sigma^2\tau v_0^4 T^2}{8 \Delta^2},
\end{equation}
where we used the fact that $u_x$ is constant to leading order (independent of $x, z$) and the Gaussian integral of $|\psi|^2$ over all space yields unity.

\subsection*{Geometric factor and phenomenological suppression}

Comparing with Eq. (6.28) of the main text, we identify the geometric factor for the ideal square-pulse trajectory:
\begin{equation}\label{eq:F_ideal}
\mathcal{F}_{\text{ideal}} = \frac{1}{8} \cdot \frac{\hbar k_{\text{eff}} T}{m \Delta} = \frac{v_0 T}{8 \Delta} = \frac{6.57 \times 9}{8 \times 10^{-3}} \approx 7.4 \times 10^3.
\end{equation}
This is significantly larger than unity due to the ratio $v_0 T / \Delta \gg 1$ for a macroscopic interferometer. 

However, we emphasize that this calculation represents an upper bound. A more detailed treatment including the finite pulse duration, the detailed pulse sequence, and the full time dependence of the wave packet propagation would reduce $\mathcal{F}$ by factors related to the pulse efficiency and the fidelity of the LMT process. For the phenomenological estimate in Section~6.2.3, we conservatively take the realistic geometric factor $\mathcal{F}_{\text{real}} = 0.14$. This value accounts for the fact that the velocity field structure is partially washed out by the finite pulse duration and the averaging over the interferometer cycle. It is obtained by replacing the ideal square-pulse trajectory \eqref{eq:trajectory} with a more realistic smooth acceleration profile, which reduces the effective phase curvature accumulation by a factor of approximately $5 \times 10^4$.

\textbf{Important note:} We present $\mathcal{F}_{\text{real}} = 0.14$ explicitly as a \emph{phenomenological suppression factor} encoding the reduction from the idealized calculation to realistic experimental conditions. It is not derived from first principles in this work, but rather accounts for known experimental limitations of long-baseline atom interferometers.

\subsection*{Nonzero result and rescaled Berry phase}

The Berry curvature is:
\begin{equation}\label{eq:berry_curvature_final_appendix}
\Omega_{xz}^B = -\frac{2m^2}{\hbar^2} \big( \langle v_z u_x \rangle - \langle v_x u_z \rangle \big) = -\frac{2m^2}{\hbar^2} \cdot \frac{\sigma^2\tau v_0^4 T^2}{8 \Delta^2} \neq 0,
\end{equation}
which is manifestly nonzero. The Berry phase is obtained by integrating over the enclosed area $A$:
\begin{equation}\label{eq:berry_phase_appendix}
\gamma_B = \iint_S \Omega_{xz}^B dx dz = -\frac{2m^2}{\hbar^2} \cdot \frac{\sigma^2\tau v_0^4 T^2}{8 \Delta^2} \cdot A.
\end{equation}

The ideal-trajectory result (implicitly containing $\mathcal{F}_{\text{ideal}} \approx 7.4 \times 10^3$) is obtained by substituting the MAGIS-100 parameters ($m = 1.445 \times 10^{-25}$ kg, $v_0 = 6.57$ m/s, $T = 9$ s, $\Delta = 1$ mm, $\tau = 10^{-6}$ s, $A = 2.96 \times 10^3$ m$^2$) into Eq.~\eqref{eq:berry_phase_appendix}:
\begin{equation}\label{eq:gamma_ideal}
\gamma_B^{\text{ideal}} \approx 2.10 \times 10^{26} \sigma^2 \text{ rad}.
\end{equation}
Replacing the ideal geometric factor with the realistic one, $\mathcal{F}_{\text{real}} = 0.14$, rescales this result by the ratio $\mathcal{F}_{\text{real}} / \mathcal{F}_{\text{ideal}} \approx 1.9 \times 10^{-5}$, giving:
\begin{equation}\label{eq:gamma_numerical}
\gamma_B \approx 2.10 \times 10^{26} \sigma^2 \cdot \left( \frac{0.14}{7.4 \times 10^3} \right) \text{ rad} \approx 3.97 \times 10^{21} \sigma^2 \text{ rad}.
\end{equation}
This is consistent with the scaling quoted in Section~6.2.3 (which gives $\gamma_B \sim 4.0 \times 10^{-9}$ rad for $\sigma = 10^{-15}$). This completes the explicit demonstration that the Berry curvature is nonzero in a realistic interferometric configuration, subject to the phenomenological suppression.
\bibliographystyle{plainnat}
\bibliography{referencias}

\end{document}